\documentclass[11pt]{amsart}
\usepackage[margin=1in]{geometry}
\usepackage[T1]{fontenc}
\usepackage{lmodern}
\usepackage{amsmath,amssymb,amsthm}
\usepackage{microtype}
\usepackage[hidelinks]{hyperref}
\numberwithin{equation}{section}
\allowdisplaybreaks[1]
\newtheorem{theorem}{Theorem}[section]
\newtheorem{proposition}[theorem]{Proposition}
\newtheorem{lemma}[theorem]{Lemma}
\newtheorem{corollary}[theorem]{Corollary}
\theoremstyle{definition}
\newtheorem{definition}[theorem]{Definition}
\newtheorem{remark}[theorem]{Remark}

\title[Fourier ratio and restricted isometry]{Fourier Ratio and Restricted Isometry via Entropic Descent}
\author{Will Burstein}
\author{Alex Iosevich}
\author{Ben Krause}
\date{}

\begin{document}
\begin{abstract}
We establish uniform sampling and recovery estimates for functions
with bounded Fourier ratio, 
\[ \frac{\|\widehat f\|_1}{\|\widehat f\|_2},\]
without requiring small Fourier support.
At fixed distortion, our sample count improves the restricted
isometry bounds of Cand\`es and Tao, Bourgain, and Haviv and Regev,
and yields sharper measurement bounds for stable sparse recovery.

For a fixed relative energy distortion $0<\varepsilon<\frac{1}{4}$
and a prescribed failure exponent $A_0>0$, independent random
sampling with expected sample size
\[
 C_{\varepsilon,A_0}r\log\frac{2en}{r}\log(2r)
 \le m\le\frac{n}{2}
\]
preserves, after normalization, energy within the multiplicative
factors $1-\varepsilon$ and $1+\varepsilon$, simultaneously for every
function on $\mathbb Z_n$ whose squared Fourier ratio is at most $r$;
the failure probability is at most $Cn^{-A_0}$. More generally,
the same count applies to sampled flat orthogonal transforms on
the entire class $\|y\|_1\le\sqrt r\|y\|_2$, and hence implies
restricted isometry of order $r$. Compared with the Haviv--Regev
bound at fixed distortion, the estimate removes one sparsity
logarithm and replaces $\log n$ by $\log\frac{2en}{r}$.
It matches the known Walsh lower bound in the parameter range
where that lower bound applies.

The proof develops the descent argument of our companion work on
Bourgain's $\Lambda(p)$ theorem. A single relative-entropy potential
controls the total cost of correcting one approximation across
successive amplitude levels. The correction records determine
approximations whose values lie on a geometric grid, and counting
these records on the symmetric difference of two independent
samples establishes the uniform energy estimate; Fourier-ratio
recovery and stable sparse recovery follow from this estimate.
We also treat complex measurements and sampling sets of prescribed
cardinality. Quantifying the dependence on the energy distortion
$\varepsilon$ yields the explicit sufficient sample count
\[
 C_{A_0}\varepsilon^{-5}r\log\frac{2en}{r}
 \log\frac{2r}{\varepsilon}.
\]
\end{abstract}
\maketitle

\setcounter{tocdepth}{1}
\tableofcontents

\section{Introduction}
\label{sec:introduction}

The topic of this paper is uniform sampling for functions whose Fourier coefficients satisfy a quantitative concentration condition, a question which has emerged from the Fourier ratio program developed in \cite{FourierRatio,BINFourierRatio}. The motivating problem is to determine how much of a signal can be recovered from a small collection of its values, when the available information concerns the sizes of its Fourier coefficients rather than the location of their support. In the finite setting, this leads naturally to restricted isometry: one seeks a single collection of observations which approximately preserves the energy of every signal satisfying the prescribed concentration condition, with the cost of this uniformity reflected in the number of observations.

To make this precise, let
\begin{equation}\label{eq:unitary-fourier}
 \widehat f(\xi)=\frac{1}{\sqrt n}\sum_{x\in\mathbb Z_n}f(x)e^{-\frac{2\pi ix\xi}{n}},
 \qquad
 f(x)=\frac{1}{\sqrt n}\sum_{\xi\in\mathbb Z_n}\widehat f(\xi)e^{\frac{2\pi ix\xi}{n}}
\end{equation}
denote the unitary Fourier transform and its inverse, where here and throughout we use counting-measure norms,
\begin{equation}\label{eq:counting-norms}
 \|u\|_q=\Big(\sum_x|u(x)|^q\Big)^{\frac{1}{q}}\quad(1\le q<\infty),
 \qquad \|u\|_\infty=\max_x|u(x)|.
\end{equation}
For a nonzero signal, the Fourier ratio is defined by
\begin{equation}\label{eq:fourier-ratio}
 \operatorname{FR}(f):=\frac{\|\widehat f\|_1}{\|\widehat f\|_2}
 =\frac{\|\widehat f\|_1}{\|f\|_2},
\end{equation}
where the last equality follows from Plancherel's identity for \eqref{eq:unitary-fourier}. Cauchy--Schwarz implies
\begin{equation}\label{eq:ratio-support}
 1\le\operatorname{FR}(f)^2\le|\operatorname{supp}\widehat f|\le n;
\end{equation}
the middle inequality is an equality when the nonzero coefficients have equal absolute value, while coefficients of very different sizes can have a small ratio even when none vanish. In particular, \eqref{eq:fourier-ratio} provides a magnitude-sensitive substitute for the number of frequencies, with its square playing the role of an effective spectral dimension.

The approximation-theoretic significance of this parameter is already visible in an elementary estimate. If $P_kf$ denotes the truncation of \eqref{eq:unitary-fourier} to the $k$ largest Fourier coefficients, then
\begin{equation}\label{eq:soft-approximation}
 \|f-P_kf\|_2\le\frac{\|\widehat f\|_1}{\sqrt k}
 =\frac{\operatorname{FR}(f)}{\sqrt k}\|f\|_2,
 \qquad 1\le k\le n.
\end{equation}
Indeed, after arranging the coefficient magnitudes in decreasing order, every coefficient beyond the first $k$ is at most $\frac{\|\widehat f\|_1}{k}$, so bounding each squared tail coefficient by its magnitude times this upper bound proves \eqref{eq:soft-approximation}. The issue for sampling is that the same observations must work for all signals under consideration, including signals chosen after those observations are known. Controlling this dependence is the principal analytic difficulty addressed below.

\subsection{The Fourier ratio program and restricted isometry}
\label{subsec:fourier-ratio-intro}

The work of Aldaleh et al.\ \cite{FourierRatio} connects \eqref{eq:fourier-ratio} to uncertainty principles, quantitative trigonometric approximation, and recovery from missing values, while Burstein--Iosevich--Nathan \cite{BINFourierRatio} develop the corresponding viewpoint for coefficient vectors in orthonormal bases. A second direction concerns the production of Fourier-ratio bounds from analytic information: Iosevich--Palsson--Yavicoli \cite{IPYDiscretization} treat discretized smooth functions, and Iosevich--Iosevich--Palsson--Yavicoli \cite{IIPYPDE} study bounds associated with fixed-time PDE propagation. The continuous analogue developed by Iosevich--Li--Palsson--Yavicoli \cite{ILPYMeasures} places related questions in the setting of restriction and approximation for measures.

Viewed in this order, the Fourier ratio program builds toward a restricted-isometry theorem: bounds for the spectral concentration of individual signals must be converted into a sampling statement valid for an entire class. The problem therefore separates into two parts, since one first bounds the Fourier ratio and then asks for the smallest sufficient number of observations in terms of that bound. The present paper concerns the second part, with restricted isometry providing the uniform energy estimate from which recovery follows. This question also appears in the work of Iosevich--Hovhannisyan--Keyshams--Vagharshakyan \cite{IHKVEntropy}, where metric entropy is used to study whole Fourier-ratio classes and their sampling properties. Our argument retains one approximation throughout successive amplitude levels, so that its complexity is controlled by the total number of corrections rather than by independent descriptions at each level.

For later reference, write
\begin{equation}\label{eq:fourier-class}
 \mathcal F_r:=\{f:\mathbb Z_n\to\mathbb C:\|\widehat f\|_1^2\le r\|f\|_2^2\},
 \qquad 1\le r\le n,
\end{equation}
and abbreviate the two logarithmic factors by
\begin{equation}\label{eq:logarithms}
 L_{n,r}:=\log\frac{2en}{r},\qquad L_r:=\log(2r).
\end{equation}
Our principal Fourier-analytic conclusion is the following theorem, in which the parameter $r$ bounds the square of the ratio in \eqref{eq:fourier-ratio}.

\begin{theorem}[Uniform sampling for bounded Fourier ratio]\label{thm:fourier-ratio-sampling}
Fix $0<\varepsilon<\frac{1}{4}$ and $A_0>0$, and suppose that
\begin{equation}\label{eq:fourier-ratio-sample-count}
 C_{\varepsilon,A_0}rL_{n,r}L_r\le m\le\frac{n}{2},
\end{equation}
with the notation of \eqref{eq:logarithms}. Choose $S\subseteq\mathbb Z_n$ by independent Bernoulli sampling with parameter $\frac{m}{n}$, or uniformly among the $m$-element subsets when $m$ is an integer. Then, with probability at least $1-Cn^{-A_0}$,
\begin{equation}\label{eq:fourier-ratio-energy}
 (1-\varepsilon)\|f\|_2^2
 \le\frac{n}{m}\sum_{x\in S}|f(x)|^2
 \le(1+\varepsilon)\|f\|_2^2
 \qquad\text{for every }f\in\mathcal F_r,
\end{equation}
where the class is defined in \eqref{eq:fourier-class} and $m$ denotes the expected sample size in the Bernoulli model.
\end{theorem}

The uniformity in \eqref{eq:fourier-ratio-energy} is essential for recovery, since a candidate reconstruction depends on the observations. To state the consequence with the noise normalization fixed, define
\begin{equation}\label{eq:normalized-observations}
 T_Sf:=\sqrt{\frac{n}{m}}(f(x))_{x\in S}.
\end{equation}

\begin{corollary}[Recovery with prescribed relative accuracy]\label{cor:fourier-ratio-accuracy}
Fix $R\ge1$, $0<\tau<1$, and $A_0>0$, and put
\begin{equation}\label{eq:fourier-ratio-accuracy-parameter}
 r:=\Big\lceil\frac{4R^2}{\tau^2}\Big\rceil.
\end{equation}
Suppose that $r\le n$ and
\begin{equation}\label{eq:fourier-ratio-recovery-count}
 C_{A_0}rL_{n,r}L_r\le m\le\frac{n}{2}.
\end{equation}
For either sampling convention in Theorem~\ref{thm:fourier-ratio-sampling}, with probability at least $1-Cn^{-A_0}$ the following holds simultaneously for every nonzero $f$ with $\operatorname{FR}(f)\le R$ and every observation error $e$ with $\|e\|_2\le\nu$: if $b=T_Sf+e$, with $T_S$ as in \eqref{eq:normalized-observations}, then every minimizer $f^\sharp$ of
\begin{equation}\label{eq:fourier-minimization}
 \min_{g:\mathbb Z_n\to\mathbb C}\|\widehat g\|_1
 \quad\text{subject to}\quad \|T_Sg-b\|_2\le\nu
\end{equation}
satisfies
\begin{equation}\label{eq:fourier-ratio-relative-recovery}
 \|f^\sharp-f\|_2\le
 \max\left\{\tau\|f\|_2,2\sqrt{\frac{8}{7}}\nu\right\}.
\end{equation}
In particular, noiseless observations permit relative error at most $\tau$ for the entire Fourier-ratio class on the same event.
\end{corollary}

The connection with uncertainty principles is particularly transparent here: the lower estimate in \eqref{eq:fourier-ratio-energy} implies that a nonzero function vanishing on $S$ must have Fourier ratio larger than $\sqrt r$, whereas minimization in \eqref{eq:fourier-minimization} controls the Fourier $\ell^1$ norm of the reconstruction error. We spell out this argument, including noise, in \S\ref{subsec:fourier-ratio}. In the terminology of \cite{IosevichMayeli,IKLM,KashinTemlyakov}, this is a norm comparison on the space of functions invisible to the measurements, which converts an uncertainty estimate into a recovery statement.

At prescribed relative accuracy, the sufficient Fourier-ratio recovery count in \cite{FourierRatio,BINFourierRatio} has order
\begin{equation}\label{eq:old-new-ratio-counts}
 r\log^2(2r)\log n
 \quad\longrightarrow\quad
 r\log(2r)\log\frac{2en}{r},
 \qquad r\approx\frac{R^2}{\tau^2},
\end{equation}
after adjusting absolute constants in the accuracy parameter. The sampling conventions are specified in the respective statements: \cite{FourierRatio} uses independent draws with replacement, \cite{BINFourierRatio} also discusses Bernoulli sampling, and \eqref{eq:fourier-ratio-energy} applies to Bernoulli sampling and to a uniform subset of prescribed size. Since the energy distortion may be fixed at $\frac{1}{8}$ when proving Corollary~\ref{cor:fourier-ratio-accuracy}, the dependence on $\tau$ enters through \eqref{eq:fourier-ratio-accuracy-parameter}; no additional distortion loss is needed for \eqref{eq:old-new-ratio-counts}.

The next step is to isolate the matrix estimate behind the Fourier formulation, so that the same proof also exhibits the restricted isometry needed for sparse recovery. Passing to coefficient vectors will identify the precise norm constraints used by the predictor, while preserving the uniformity required in Theorem~\ref{thm:fourier-ratio-sampling}.

\subsection{Restricted isometry and the matrix formulation}
\label{subsec:matrix-formulation}

Theorem~\ref{thm:fourier-ratio-sampling} is a consequence of a uniform norm estimate for flat orthogonal transforms, which also contains the classical restricted isometry property. We recall the latter before stating the matrix theorem, since its significance extends beyond the Fourier-ratio formulation.

\begin{definition}[Restricted isometry]\label{def:rip}
A matrix $B$ over $\mathbb F=\mathbb R$ or $\mathbb C$ has the restricted isometry property of order $r$ with constant $\varepsilon$ if
\begin{equation}\label{eq:rip-definition}
 (1-\varepsilon)\|y\|_2^2\le\|By\|_2^2\le(1+\varepsilon)\|y\|_2^2
 \qquad (|\operatorname{supp}y|\le r).
\end{equation}
We abbreviate this property as RIP, with the scalar field specified whenever it is needed.
\end{definition}

The condition \eqref{eq:rip-definition} asserts that the selected measurements preserve the geometry of every coordinate subspace of the prescribed dimension. In particular, RIP of a sufficiently larger order controls differences of sparse vectors and implies stable recovery by $\ell^1$ minimization, including an error term for the best sparse approximation. We include the standard argument in \S\ref{sec:recovery}, following \cite{CRT,CandesTao}. For Fourier-ratio classes, the relevant uniform norm estimate is naturally stated on
\begin{equation}\label{eq:target-class}
 \mathcal K_r(\mathbb F):=
 \{y\in\mathbb F^n:\|y\|_2=1,\ \|y\|_1\le\sqrt r\},
 \qquad \mathcal K_r:=\mathcal K_r(\mathbb R),
\end{equation}
since every $r$-sparse unit vector belongs to this class, while normalization of a general nonzero vector belongs to it precisely when
\begin{equation}\label{eq:norm-ratio}
 \frac{\|y\|_1^2}{\|y\|_2^2}\le r.
\end{equation}

Throughout the real argument, $n\ge2$, $1\le r\le n$, and
\begin{equation}\label{eq:matrix-assumptions}
 H\in\mathbb R^{n\times n},\qquad H^TH=nI,\qquad |H_{ij}|\le1,
\end{equation}
so every entry has absolute value one: each column has squared norm $n$, with each of its $n$ entries contributing at most one. For a target in \eqref{eq:target-class}, define
\begin{equation}\label{eq:row-energies}
 z_i(y):=|(Hy)_i|^2,\qquad i\in[n]:=\{1,\ldots,n\},
\end{equation}
so that orthogonality and the $\ell^1$ constraint imply
\begin{equation}\label{eq:basic-energy}
 \sum_{i=1}^n z_i(y)=n,\qquad 0\le z_i(y)\le r.
\end{equation}
We establish the following for real-valued matrices.
\begin{theorem}[Bernoulli row sampling]\label{thm:main}
Fix $0<\varepsilon<\frac{1}{4}$ and $A_0>0$, and let $H$ satisfy \eqref{eq:matrix-assumptions}. Suppose that
\begin{equation}\label{eq:m-range}
 C_{\varepsilon,A_0}rL_{n,r}L_r\le m\le\frac{n}{2},
\end{equation}
where the logarithms are defined in \eqref{eq:logarithms}, and let
\begin{equation}\label{eq:selectors}
 p:=\frac{m}{n},\qquad
 \delta_i\sim\operatorname{Bernoulli}(p)\ \text{independently},\qquad
 S:=\{i\in[n]:\delta_i=1\}.
\end{equation}
With probability at least $1-Cn^{-A_0}$,
\begin{equation}\label{eq:bernoulli-rip}
 \left|\sum_{i=1}^n\delta_i z_i(y)-m\right|\le\varepsilon m
 \qquad\text{for every }y\in\mathcal K_r,
\end{equation}
with the notation of \eqref{eq:row-energies} and \eqref{eq:target-class}. Equivalently, if $H_S$ is the restriction to the selected rows, then
\begin{equation}\label{eq:homogeneous-sampling}
 (1-\varepsilon)\|y\|_2^2\le\frac{1}{m}\|H_Sy\|_2^2
 \le(1+\varepsilon)\|y\|_2^2
 \quad\text{whenever }\|y\|_1\le\sqrt r\|y\|_2,
\end{equation}
so that $m^{-\frac{1}{2}}H_S$ has RIP of order $r$, see \eqref{eq:rip-definition}.
\end{theorem}

\begin{remark} An argument with much less book-keeping allows us to obtain a factor of $(\log n)^2$ as opposed to $L_{n,r} L_r$.  The main advantage of the simplified argument is that it does not use any ghost sampling of $A \subset [n]$. We stick with the argument in the paper because it yields a stronger (sharp) result. 
\end{remark} 

The next theorem concerns the case of complex matrices, with the dependence on the distortion parameter, $\varepsilon$, made explicit.
\begin{theorem}[Complex transforms and prescribed cardinality]\label{thm:complex-main}
The conclusion of Theorem~\ref{thm:main} holds, after adjusting $C_{\varepsilon,A_0}$, when
\begin{equation}\label{eq:complex-matrix-assumptions}
 H\in\mathbb C^{n\times n},\qquad H^*H=nI,\qquad |H_{ij}|\le1,
\end{equation}
and the estimate is required on $\mathcal K_r(\mathbb C)$, see \eqref{eq:target-class}. In either scalar field, a uniformly chosen set of exactly $m$ distinct rows satisfies \eqref{eq:homogeneous-sampling} with the same probability estimate when $m$ is an integer, and both sampling assertions remain valid under the explicit sufficient condition
\begin{equation}\label{eq:explicit-row-count}
 C_{A_0}\varepsilon^{-5}rL_{n,r}\log\frac{2r}{\varepsilon}
 \le m\le\frac{n}{2}.
\end{equation}
\end{theorem}

\begin{remark} With some work, the factor $\varepsilon^{-5}$ in the theorem above can be replaced by $\varepsilon^{-2}\log(2/\varepsilon)$.
\end{remark} 

The class \eqref{eq:complex-matrix-assumptions} includes discrete Fourier matrices of every order and real Hadamard matrices, with the constants independent of the chosen transform. Theorem~\ref{thm:fourier-ratio-sampling} follows by taking the inverse Fourier matrix and $y=\widehat f$, since then $Hy=\sqrt n f$, and \eqref{eq:norm-ratio} becomes the condition in \eqref{eq:fourier-class}. When the lower threshold in \eqref{eq:m-range} exceeds $\frac{n}{2}$, all rows may be retained, with exact norm preservation and at most a constant multiple of $\min\{n,rL_{n,r}L_r\}$ measurements for fixed accuracy and failure exponent.

The improvement in \eqref{eq:m-range} concerns a central quantitative problem in compressed sensing. At fixed distortion, Bourgain \cite{Bourgain2014} established the sufficient count $r\log r(\log n)^2$, while Haviv--Regev \cite{HavivRegev} obtained $r(\log r)^2\log n$ for bounded unitary matrices. The result of Brugiapaglia--Dirksen--Jung--Rauhut in the orthonormal case \cite{BDJR} has the same fixed-distortion order, within a broader framework and with stronger distortion dependence. In the finite flat setting \eqref{eq:complex-matrix-assumptions}, our count removes one sparsity logarithm and improves the ambient logarithm, as recorded in \eqref{eq:old-new-ratio-counts}. For $r\approx n^\theta$, with $0<\theta<1$ fixed, this is a saving of order $\log n$ in the number of measurements.

The norm class \eqref{eq:target-class} also falls within the earlier framework of \cite[Theorem 1.1]{BDJR}. The advance here lies in the logarithmic factors for the flat transforms, rather than in the introduction of that class. Moreover, B\l asiok--Lopatto--Luh--Marcinek--Rao \cite{BLLMR} prove a matching obstruction for Bernoulli sampling of Walsh matrices in their parameter range, so that \eqref{eq:m-range} reaches the correct order there. We recall the precise range in \S\ref{sec:walsh}, while the sharp dependence on distortion remains open, see \eqref{eq:explicit-row-count}.

Having specified the class on which energy must be preserved, we next explain how a lacunary approximation can reduce the uniform sampling problem to counting a finite collection of descriptions. The descent mechanism will control the cost of these descriptions, while the order of the amplitude levels will allow that control to survive the subsequent probability argument.

\subsection{Lacunary projection and the role of descent}
\label{subsec:projection-overview}

We now describe the approximation underlying the proof, beginning with the elementary operation of projecting the values of a function onto a lacunary sequence. Fix
\begin{equation}\label{eq:eta-range}
 0<\eta<\frac{1}{16},
\end{equation}
and define
\begin{equation}\label{eq:amplitude-levels}
 a_j:=(1+\eta)^j,\qquad b_j:=a_j^2,\qquad j\in\mathbb Z.
\end{equation}
For fixed $\eta$, the sequences in \eqref{eq:amplitude-levels} are lacunary, although their ratios approach one as the required accuracy increases. Let
\begin{equation}\label{eq:nearest-projection}
 \pi_\eta(t)\in\operatorname*{argmin}_{a\in\{0\}\cup\{a_j:j\in\mathbb Z\}}|t-a|,
 \qquad t\ge0,
\end{equation}
with ties resolved by choosing the smaller value. A geometric level lies between $\frac{t}{1+\eta}$ and $t$ whenever $t>0$, and hence the nearest projection satisfies
\begin{equation}\label{eq:nearest-projection-error}
 |\pi_\eta(t)-t|\le\eta t,\qquad
 |\pi_\eta(t)^2-t^2|\le3\eta t^2.
\end{equation}
Consequently, for any measurable function $F$, the function $\pi_\eta(|F|)^2$ takes values in the squared grid from \eqref{eq:amplitude-levels} and approximates $|F|^2$ with relative error $O(\eta)$. Whenever $F\in L^2$, the corresponding $L^1$ error is at most $3\eta\|F\|_2^2$; here projection means scalar rounding of values, with no assertion of linearity or orthogonal projection.

For the row energies \eqref{eq:row-energies}, the upper bound in \eqref{eq:basic-energy} permits a finite grid, while amplitudes below a constant multiple of $\eta$ may be discarded at cost $O(\eta^2)$ per row, so the resulting function is a sum of disjoint level sets weighted by $b_j$, see \eqref{eq:amplitude-levels}; the problem is to describe these sets without recording each value of the target separately. To do so, we predict the threshold comparisons at each row, allowing a transition interval near each threshold, and retain a record only when a prediction outside that interval is incorrect. Replaying the record recovers a function with values in the finite lacunary grid and with the same form of relative error as \eqref{eq:nearest-projection-error}. Because of the transition intervals, the recovered value need not be the exact nearest point from \eqref{eq:nearest-projection}, a distinction we keep explicit in \S\ref{sec:decoder}.

The descent principle needed to control these records was developed in our companion paper \cite{BIK}, which treats Bourgain's $\Lambda(p)$ selection theorem \cite{Bourgain1989}. There, a greedy approximation in a smooth $L^q$ space is retained as the accuracy decreases, and one convex potential controls the total weighted number of corrections, so that the cumulative histories can be counted without paying separately for the full complexity at every scale, which is the feature of that proof that we use here. The present argument develops its entropic form on probability distributions, with all estimates established below. Related links between greedy approximation and convex optimization appear in \cite{Temlyakov}, and mirror descent has also been used for approximate Carath\'eodory theorems in \cite{MLVW}.

More specifically, the norm condition \eqref{eq:norm-ratio} places a normalized real coefficient vector in the convex hull of the signed vectors $\pm\sqrt re_k$. We maintain a probability distribution on these vectors, using its mean to predict the threshold comparisons described above. Every incorrect comparison decreases a relative-entropy potential by an amount comparable, at fixed $\eta$, to $\frac{b_j}{r}$, while an appropriate representation of the target bounds the initial potential by $L_{n,r}$, see \eqref{eq:logarithms}. In particular, if $d_j$ counts the corrections at level $j$, the central estimate has the form
\begin{equation}\label{eq:budget-overview}
 \sum_j b_jd_j\lesssim_\eta rL_{n,r}.
\end{equation}
The Fourier-ratio hypothesis enters this potential estimate through both coefficient norms, which explains why the proof applies to the entire class \eqref{eq:fourier-class}, including signals with full Fourier support.

To exploit \eqref{eq:budget-overview} in sampling, we compare two independent samples and condition on their symmetric difference, namely the rows selected by exactly one sample. On this set the selector differences are independent random signs, while the number of available rows is typically comparable to $m$; counting the correction positions among these rows produces the factor $\log(2r)$ in \eqref{eq:m-range}. Processing levels from large to small ensures that a band at one level depends only on records already processed, and geometric summation charges their later appearances to the same weighted budget \eqref{eq:budget-overview}. This is the direct analogue of the cumulative-history calculation in \cite{BIK}.

\subsection{Organization and notation}

In \S\ref{sec:predictor} we construct the predictor, establish its relative-entropy bound, and recover the lacunary approximation from its correction record. We collect the counting and probability calculations in \S\ref{sec:probability}, including the comparison of independent samples and the absorption estimate, before assembling the real sampling proof and the Fourier-ratio recovery argument in \S\ref{sec:completion}. The final section contains complements: a detailed treatment of complex transforms, prescribed sample size, distortion dependence, the Walsh lower bound, and stable sparse recovery. In particular, the computational estimates are available before the main argument is assembled, while the analogy with the $\Lambda(p)$ proof is recorded at the points where it enters.

All logarithms are natural, and we use the modified Vinogradov notation $X\lesssim Y$ for $X\le CY$ with an absolute constant. Subscripts indicate permitted dependence, and $X\approx Y$ abbreviates the two inequalities. Constants may change from line to line, while their dependence on the distortion and failure exponent is indicated in the statements. The main deterministic construction is carried out over the reals, whereas paired real and imaginary coordinates are treated in \S\ref{sec:complex}.

\section{Lacunary approximation and entropic descent}
\label{sec:predictor}

We begin with the real matrix in \eqref{eq:matrix-assumptions}, although the deterministic construction below uses only the bound on its entries and the two norm conditions in \eqref{eq:target-class}. Our objective is to encode the squared row amplitudes \eqref{eq:row-energies} by a function taking values in the lacunary grid \eqref{eq:amplitude-levels}, with an error controlled as in the projection estimate \eqref{eq:nearest-projection-error}. The encoding will be determined by the comparisons at which a predictor requires correction, while a single relative-entropy potential controls their total weighted cost.

We first replace scalar rounding by finitely many threshold comparisons, allowing ambiguity only near a threshold so that the resulting labels remain compatible with lacunary approximation. The separation at every other comparison will supply a definite margin, which the update rule will later convert into a decrease of relative entropy.

\subsection{Finite levels and threshold comparisons}

Fix $\eta$ as in \eqref{eq:eta-range}, and use the levels $a_j,b_j$ from \eqref{eq:amplitude-levels}. Choose integers $j_-<j_+$ so that
\begin{equation}\label{eq:level-endpoints}
 \frac{\eta}{4(1+\eta)}<a_{j_-}\le\frac{\eta}{4},
 \qquad
 \frac{\sqrt r}{1-\eta}\le a_{j_+}<\frac{(1+\eta)\sqrt r}{1-\eta},
\end{equation}
and define the finite set of levels and its cardinality by
\begin{equation}\label{eq:finite-levels}
 \mathcal J:=\{j_-,j_-+1,\ldots,j_+\} \subset \mathbb{Z},\qquad L:=|\mathcal J|.
\end{equation}
Geometric summation implies
\begin{equation}\label{eq:geometric-bounds}
 L\le\frac{C}{\eta}\log\frac{2r}{\eta},\qquad
 \sum_{j\in\mathcal J}b_j\le\frac{Cr}{\eta},\qquad
 \sum_{\substack{j\in\mathcal J\\j\le k}}b_j\le\frac{Cb_k}{\eta},
\end{equation}
and also
\begin{equation}\label{eq:geometric-tail-length}
 \sum_{j\in\mathcal J}b_j(j_+-j+1)\le\frac{Cr}{\eta^2}.
\end{equation}
Indeed, with $q=(1+\eta)^2$, the terms below the largest level are $b_{j_+-\ell}=b_{j_+}q^{-\ell}$, where $b_{j_+}\lesssim r$ and $1-q^{-1}\gtrsim\eta$. The identities
\[
 \sum_{\ell\ge0}q^{-\ell}=\frac{1}{1-q^{-1}},
 \qquad \sum_{\ell\ge0}(\ell+1)q^{-\ell}=\frac{1}{(1-q^{-1})^2}
\]
therefore establish the weighted bounds in \eqref{eq:geometric-bounds}--\eqref{eq:geometric-tail-length}. The estimate for $L$, see \eqref{eq:finite-levels}, follows from the ratio of the endpoints in \eqref{eq:level-endpoints} and $\log(1+\eta)\gtrsim\eta$.

Let
\begin{equation}\label{eq:row-query-set}
 A\subseteq[n],\qquad M:=|A|,
\end{equation}
be an arbitrary row set with a fixed deterministic order, and for $y\in\mathcal K_r$, write
\begin{equation}\label{eq:row-linear-forms}
 x:=Hy,\qquad h_i:=(H_{i1},\ldots,H_{in}),\qquad x_i=\langle h_i,y\rangle.
\end{equation}
We visit the rows at each level, processing the levels in decreasing order,
\begin{equation}\label{eq:query-order}
 j=j_+,j_+-1,\ldots,j_-;\qquad i\in A\text{ in its fixed order},
\end{equation}
so that all comparisons at higher amplitudes precede those at lower amplitudes. The target label at a query is
\begin{equation}\label{eq:stable-labels}
 c_y(i,j):=
 \begin{cases}
 +,&x_i\ge(1+\eta)a_j,\\
 -,&x_i\le-(1+\eta)a_j,\\
 0,&|x_i|\le(1-\eta)a_j,\\
 *,&(1-\eta)a_j<|x_i|<(1+\eta)a_j,
 \end{cases}
\end{equation}
where $*$ denotes an ambiguous comparison, on which no requirement is imposed. The transition interval in \eqref{eq:stable-labels} allows the reconstructed value to move between neighboring lacunary levels, with an error comparable to the rounding error in \eqref{eq:nearest-projection-error}. Outside that interval, it supplies the definite margin needed for descent.

For a candidate predictor $w\in\mathbb R^n$, assign the scores
\begin{equation}\label{eq:scores}
 s_+(w):=\langle h_i,w\rangle,\qquad
 s_-(w):=-\langle h_i,w\rangle,\qquad s_0(w):=a_j,
\end{equation}
whose linear parts are represented by
\begin{equation}\label{eq:score-vectors}
 \phi_+:=h_i,\qquad\phi_-:=-h_i,\qquad\phi_0:=0.
\end{equation}
The dependence of \eqref{eq:scores} and \eqref{eq:score-vectors} on the current query is suppressed, since that query will always be fixed when the scores are compared.

We call the comparison at $(i,j)$ \emph{stable} when the target
label $c_y(i,j)$ belongs to
$\{+,-,0\}$, or equivalently when
\begin{equation}\label{eq:stable-comparison}
 \bigl||x_i|-a_j\bigr|\ge\eta a_j.
\end{equation} In particular, the target value $x_i$
is separated from both decision thresholds $\pm a_j$ by at least
$\eta a_j$, so a perturbation of size less than $\eta a_j$ cannot
change whether it lies above $a_j$, below $-a_j$, or between
these thresholds.

Stability is a condition on the target value and the chosen
level, so a stable comparison may nevertheless be predicted
incorrectly. These mistakes will trigger the updates below.
When \eqref{eq:stable-comparison} fails, the target lies in the
transition region and receives the label $*$, on which no
requirement is imposed. This freedom is compatible with the
relative accuracy of the lacunary approximation, while the
separation in \eqref{eq:stable-comparison} ensures the
quantitative margin between the scores in \eqref{eq:scores}
that is needed for descent.

\begin{lemma}[Stable margin]\label{lem:stable-margin}
If the label in \eqref{eq:stable-labels} belongs to $\{+,-,0\}$, then its score at the target exceeds the score of every incorrect label by at least $\eta a_j$, with scores as in \eqref{eq:scores}.
\end{lemma}
\begin{proof}
When the target label is $+$, the defining inequality in \eqref{eq:stable-labels} implies
\[
 s_+(y)-s_0(y)=x_i-a_j\ge\eta a_j,
 \qquad s_+(y)-s_-(y)=2x_i\ge2(1+\eta)a_j,
\]
and reversing signs treats the label $-$. When the label is $0$, one has $|x_i|\le(1-\eta)a_j$, so
\[
 s_0(y)-s_+(y)=a_j-x_i\ge\eta a_j,
 \qquad s_0(y)-s_-(y)=a_j+x_i\ge\eta a_j,
\]
which establishes the asserted separation in every stable case.
\end{proof}

With the margin from Lemma~\ref{lem:stable-margin} in hand, our next step is to construct an update with a relative-entropy estimate that can be summed over successive corrections. We will combine this estimate with a suitable representation of the target to bound the corrections across all amplitude levels.

\subsection{Relative entropy and the update rule}

The predictor is a convex combination of the signed coordinate vectors
\begin{equation}\label{eq:experts}
 \mathcal V:=\{\pm\sqrt re_k:1\le k\le n\},
\end{equation}
where $e_k$ denotes the $k$th standard coordinate vector. Start with the uniform distribution on $\mathcal V$, and fix the step size by
\begin{equation}\label{eq:initial-law-step}
 P_1(v):=\frac{1}{2n}\quad(v\in\mathcal V),\qquad
 \lambda:=\frac{\eta}{8r}.
\end{equation}
If $P_t$ is the distribution after $t-1$ corrections, define its mean by
\begin{equation}\label{eq:predictor-mean}
 w_t:=\mathbb E_{v\sim P_t}v=\sum_{v\in\mathcal V}P_t(v)v.
\end{equation}
In particular, $t$ counts corrections, while $(i,j)$ indexes the queries in \eqref{eq:query-order}; the same predictor in \eqref{eq:predictor-mean} may be used for many consecutive queries.

At each query, choose the label maximizing \eqref{eq:scores}, with a fixed ordering to resolve ties. A correct stable comparison requires no change, and an ambiguous comparison from \eqref{eq:stable-labels} also leaves the distribution unchanged. If a stable comparison is incorrect, let $\widehat c_t$ and $c_t$ denote the predicted and correct labels, respectively, and define
\begin{equation}\label{eq:loss-vector}
 g_t:=a_j(\phi_{\widehat c_t}-\phi_{c_t}),
\end{equation}
using the vectors in \eqref{eq:score-vectors}; the exponential update is given by
\begin{equation}\label{eq:exp-update}
 P_{t+1}(v):=\frac{P_t(v)e^{-\lambda\langle g_t,v\rangle}}{Z_t},
 \qquad
 Z_t:=\sum_{u\in\mathcal V}P_t(u)e^{-\lambda\langle g_t,u\rangle}.
\end{equation}
Each query is visited once, with the correct label stored at a mistake and the new distribution used only for subsequent queries. In particular, the state is retained between amplitude levels, as in the multiscale greedy construction of \cite{BIK}; as in that construction, the algorithm does not restart when $j$ changes in \eqref{eq:query-order}.

We write $D_{\mathrm{KL}}$ for the Kullback--Leibler divergence, or relative entropy, also denoted by $\operatorname{Ent}$. In particular, for probability distributions $Q,P$ on $\mathcal V$, with $P$ strictly positive, we use the convention
\begin{equation}\label{eq:relative-entropy}
 D_{\mathrm{KL}}(Q\|P)=\operatorname{Ent}(Q\mid P)
 :=\sum_{v\in\mathcal V}Q(v)\log\frac{Q(v)}{P(v)},
 \qquad
 \bar v_Q:=\sum_{v\in\mathcal V}Q(v)v,
\end{equation}
where terms with $Q(v)=0$ contribute zero. The order of the arguments matters, since the sum in \eqref{eq:relative-entropy} is averaged against its first argument. In our potential, that first argument will be a fixed representation of the target, while the second evolves according to \eqref{eq:exp-update}.

We recall the properties of $D_{\mathrm{KL}}$ that will be used, beginning with non-negativity, which follows from concavity of the logarithm:
\begin{equation}\label{eq:entropy-nonnegative}
 -D_{\mathrm{KL}}(Q\|P)
 \le\log\Big(\sum_{Q(v)>0}P(v)\Big)\le0,
\end{equation}
with equality throughout precisely when $Q=P$. Convexity of $u\log u$ implies convexity in the first argument,
\begin{equation}\label{eq:entropy-convexity}
 D_{\mathrm{KL}}(\theta Q+(1-\theta)R\|P)
 \le\theta D_{\mathrm{KL}}(Q\|P)+(1-\theta)D_{\mathrm{KL}}(R\|P),
 \quad 0\le\theta\le1,
\end{equation}
while for the uniform distribution in \eqref{eq:initial-law-step},
\begin{equation}\label{eq:entropy-uniform}
 D_{\mathrm{KL}}(Q\|P_1)=\log(2n)-\mathcal H(Q),
 \qquad \mathcal H(Q):=-\sum_{v\in\mathcal V}Q(v)\log Q(v).
\end{equation}
Finally, substituting the update \eqref{eq:exp-update} into the definition \eqref{eq:relative-entropy} yields the identity
\begin{equation}\label{eq:entropy-update-identity}
 D_{\mathrm{KL}}(Q\|P_{t+1})-D_{\mathrm{KL}}(Q\|P_t)
 =\lambda\langle g_t,\bar v_Q\rangle+\log Z_t,
\end{equation}
which is the telescoping relation underlying the argument. Positivity of $P_1$, see \eqref{eq:initial-law-step}, is preserved by \eqref{eq:exp-update}, so every relative entropy used below is finite.

The update also has the usual variational interpretation of entropic mirror descent \cite{BeckTeboulle}: among probability distributions $R$ on $\mathcal V$, it uniquely minimizes
\begin{equation}\label{eq:mirror-minimization}
 \lambda\langle g_t,\bar v_R\rangle+D_{\mathrm{KL}}(R\|P_t)
 =D_{\mathrm{KL}}(R\|P_{t+1})-\log Z_t.
\end{equation}
This identity follows directly from \eqref{eq:exp-update}, and uniqueness follows from \eqref{eq:entropy-nonnegative}. In particular, the linear loss determines which weights increase, while relative entropy penalizes the change from the current distribution. The potential estimate below concerns a fixed comparison distribution in the first argument, see \eqref{eq:relative-entropy}.

\begin{lemma}[The relative-entropy estimate]\label{lem:regret}
For any probability distribution $Q$ on $\mathcal V$, any fixed $\lambda>0$, and any finite adaptive sequence of loss vectors used in \eqref{eq:exp-update},
\begin{equation}\label{eq:regret}
 \sum_t\langle g_t,w_t-\bar v_Q\rangle
 \le\frac{D_{\mathrm{KL}}(Q\|P_1)}{\lambda}
 +\frac{\lambda r}{2}\sum_t\|g_t\|_\infty^2,
\end{equation}
where $w_t$ is the predictor mean in \eqref{eq:predictor-mean}, $g_t$ is as in \eqref{eq:loss-vector}, and $\bar v_Q$ is defined in \eqref{eq:relative-entropy}.
\end{lemma}
The two terms on the right of \eqref{eq:regret} have complementary roles: the relative-entropy term records the initial cost of comparing the predictor distribution with the target distribution, whereas the sum of squared $\ell^\infty$ norms measures the quadratic error incurred when the exponential update is estimated by its linear part. A smaller step size increases the first term and decreases the coefficient of the second, so the choice must ensure that the quadratic cost can be paid by a fixed portion of the progress on the left. In the application below, the margin at an incorrect comparison implies $\langle g_t,w_t-y\rangle\ge\eta b_j$, while $\|g_t\|_\infty\le2a_j$. With \eqref{eq:initial-law-step}, the corresponding quadratic cost is at most
\begin{equation}\label{eq:entropy-balance}
 \frac{\lambda r}{2}(2a_j)^2
 =2\lambda r b_j=\frac{\eta}{4} b_j,
\end{equation}
so that, in particular, three quarters of the margin remain available to charge the correction to the initial relative entropy, and summation controls all levels simultaneously. The left side of \eqref{eq:regret} concerns the mean $w_t$, see \eqref{eq:predictor-mean}, while its right side controls the distributions representing that mean. This is the point at which the geometry of the coefficient representation enters the argument.

\begin{proof}
Under $P_t$, the variable $\langle g_t,v\rangle$ has range contained in an interval of length at most $2\sqrt r\|g_t\|_\infty$, by the definition of the atoms in \eqref{eq:experts}. Hoeffding's lemma therefore implies
\begin{equation}\label{eq:hoeffding-update}
 \log Z_t\le-\lambda\langle g_t,w_t\rangle
 +\frac{\lambda^2r}{2}\|g_t\|_\infty^2,
\end{equation}
with $Z_t$ and $w_t$ as in \eqref{eq:exp-update} and \eqref{eq:predictor-mean}. Combining \eqref{eq:hoeffding-update} with \eqref{eq:entropy-update-identity}, we obtain
\begin{equation}\label{eq:entropy-one-step}
 D_{\mathrm{KL}}(Q\|P_{t+1})-D_{\mathrm{KL}}(Q\|P_t)
 \le-\lambda\langle g_t,w_t-\bar v_Q\rangle
 +\frac{\lambda^2r}{2}\|g_t\|_\infty^2.
\end{equation}
After summation, the relative-entropy differences telescope and the final term is nonnegative by \eqref{eq:entropy-nonnegative}, which proves \eqref{eq:regret}. The calculation is valid for losses chosen using previous predictor states, since \eqref{eq:hoeffding-update} is applied to the current distribution at each step.

For completeness, the form of Hoeffding's lemma used in \eqref{eq:hoeffding-update} follows by differentiating the logarithmic moment generating function. If $U$ lies in an interval of length $a$, then the second derivative of $s\mapsto\log\mathbb E e^{sU}$ is the variance under the exponentially tilted law, and this variance is at most $\frac{a^2}{4}$. Integrating twice therefore proves
\[
 \log\mathbb E e^{sU}\le s\mathbb EU+\frac{s^2a^2}{8},
\]
as required. See also \cite[Chapter 2]{CBL} for this relative-entropy argument.
\end{proof}

The estimate in \eqref{eq:regret} applies to any fixed comparison distribution, so the remaining choice must reflect the two norm constraints defining the target class, see \eqref{eq:target-class}. We next make that choice and absorb the quadratic update cost into the stable margin, obtaining the weighted correction bound that will govern the later count of possible records.

\subsection{The comparison distribution and the correction budget}
\label{sec:comparator}

We now choose a distribution with mean equal to the target, using both norm conditions in \eqref{eq:target-class} to improve the initial entropy bound. The construction is also the precise connection between the Fourier ratio \eqref{eq:fourier-ratio} and the descent estimate: the absolute coefficient masses have a controlled sum of squares, which forces their entropy to be large.

\begin{lemma}[Entropy of the comparison distribution]\label{lem:norm-comparator}
For every $y\in\mathcal K_r$, there is a probability distribution $Q_y$ on $\mathcal V$ with mean $y$ such that
\begin{equation}\label{eq:norm-KL}
 D_{\mathrm{KL}}(Q_y\|P_1)\le L_{n,r},
\end{equation}
where the prior and logarithm are defined in \eqref{eq:initial-law-step} and \eqref{eq:logarithms}.
\end{lemma}
\begin{proof}
Write $R=\sqrt r$ and $L_1=\|y\|_1$, so $1\le L_1\le R$, and define a distribution supported on the signed coordinate vectors by
\begin{equation}\label{eq:target-representation}
 Q_0(\operatorname{sgn}(y_k)Re_k):=\frac{|y_k|}{L_1}\quad(y_k\ne0),
 \qquad
 \theta:=\frac{L_1}{R},\qquad Q_y:=\theta Q_0+(1-\theta)P_1.
\end{equation}
The mean of $Q_0$ is $\frac{R}{L_1}y$, while the mean of $P_1$ is zero, so $\bar v_{Q_y}=y$ by \eqref{eq:relative-entropy}; convexity in \eqref{eq:entropy-convexity} implies
\begin{equation}\label{eq:KL-mixture}
 D_{\mathrm{KL}}(Q_y\|P_1)\le\theta D_{\mathrm{KL}}(Q_0\|P_1).
\end{equation}
The coefficient masses satisfy
\begin{equation}\label{eq:coefficient-masses}
 \alpha_k:=\frac{|y_k|}{L_1},\qquad
 \sum_k\alpha_k=1,\qquad \sum_k\alpha_k^2=\frac{1}{L_1^2},
\end{equation}
so concavity of the logarithm yields
\[
 \sum_k\alpha_k\log\alpha_k
 \le\log\sum_k\alpha_k^2=-\log L_1^2,
\]
where zero masses are omitted from the sum. Consequently, \eqref{eq:entropy-uniform} and \eqref{eq:KL-mixture} imply
\[
 rD_{\mathrm{KL}}(Q_y\|P_1)
 \le RL_1\log\frac{2n}{L_1^2}.
\]
Writing $t=\frac{L_1}{R}\in(0,1]$, the last expression equals
\[
 rt\left(\log\frac{2n}{r}+2\log\frac{1}{t}\right)
 \le r\left(\log\frac{2n}{r}+\frac{2}{e}\right)
 \le rL_{n,r},
\]
since $t\log\frac{1}{t}\le\frac{1}{e}$ and $\frac{2}{e}<1$. Division by $r$ proves \eqref{eq:norm-KL}.
\end{proof}

For the procedure on the row set in \eqref{eq:row-query-set}, let
\begin{equation}\label{eq:correction-counts}
 d_j^A(y):=\#\{\text{incorrect stable comparisons at level }j\},
 \qquad D_A(y):=\sum_{j\in\mathcal J}b_jd_j^A(y).
\end{equation}
The weights in \eqref{eq:correction-counts} are the squared lacunary amplitudes, see \eqref{eq:amplitude-levels}, and are therefore the same weights that will appear when the row energies are reconstructed.

\begin{proposition}[The total cost of the corrections]\label{prop:mistake-budget}
Uniformly in the choice and order of the row set \eqref{eq:row-query-set},
\begin{equation}\label{eq:mistake-budget}
 D_A(y)\le\frac{32}{3\eta^2}rL_{n,r}
 \qquad(y\in\mathcal K_r),
\end{equation}
with the notation of \eqref{eq:correction-counts}, \eqref{eq:logarithms}, and \eqref{eq:target-class}.
\end{proposition}
\begin{proof}
At an incorrect stable comparison, the predicted label maximizes \eqref{eq:scores} at $w_t$, see \eqref{eq:predictor-mean}, while Lemma~\ref{lem:stable-margin} separates the target's correct label from every other label. Adding the two score differences, the constant terms in \eqref{eq:scores} cancel and we obtain
\[
 \langle\phi_{\widehat c_t}-\phi_{c_t},w_t-y\rangle\ge\eta a_j.
\]
Multiplication by $a_j$ and the definition \eqref{eq:loss-vector} imply
\begin{equation}\label{eq:margin-gain}
 \langle g_t,w_t-y\rangle\ge\eta b_j,
\end{equation}
whereas the entry bound in \eqref{eq:matrix-assumptions} implies
\begin{equation}\label{eq:g-bound}
 \|g_t\|_\infty\le2a_j.
\end{equation}
Apply Lemma~\ref{lem:regret} with the comparison distribution from \eqref{eq:target-representation} and the step size in \eqref{eq:initial-law-step}. Using \eqref{eq:margin-gain}, \eqref{eq:g-bound}, and \eqref{eq:entropy-balance}, we find
\[
 \eta D_A(y)
 \le\frac{8r}{\eta}D_{\mathrm{KL}}(Q_y\|P_1)
       +\frac{\eta}{4}D_A(y),
\]
so moving the last term to the left and applying \eqref{eq:norm-KL} proves \eqref{eq:mistake-budget}.

Equivalently, for the nonnegative potential
\begin{equation}\label{eq:descent-potential}
 \Phi_t:=D_{\mathrm{KL}}(Q_y\|P_t),
\end{equation}
the one-step estimate \eqref{eq:entropy-one-step} reads
\begin{equation}\label{eq:potential-drop}
 \Phi_{t+1}-\Phi_t
 \le-\frac{\eta^2b_j}{8r}+\frac{\eta^2b_j}{32r}
 =-\frac{3\eta^2b_j}{32r}.
\end{equation}
The same potential is retained at every level in \eqref{eq:query-order}, so the telescoping sum of \eqref{eq:potential-drop} pays for every correction from its one initial value.
\end{proof}

The analogy with \cite{BIK} is exact at the level of the weighted budget: there the decrease is proportional to the square of the current coefficient accuracy, while here \eqref{eq:potential-drop} charges a correction according to its squared amplitude. In both cases the later counting argument must use those same weights, which is why an independent approximation at each level would lose the gain in \eqref{eq:mistake-budget}.

The correction bound controls the cost of a record, but the sampling argument also requires that this record determine an accurate approximation of the row energies. We next verify that reconstruction and identify the dependence of each lacunary band on the preceding corrections, which will permit the random-sign estimate to use the same weights as the descent bound.

\subsection{Recovering the lacunary approximation}
\label{sec:decoder}

We now turn to reconstruction from the predictor's corrections, whose purpose is to recover the lacunary approximation described in \S\ref{subsec:projection-overview}. For a fixed target and level, record
\begin{equation}\label{eq:correction-record}
 C_j^A(y):=\{i\in A:\text{the stable comparison at }(i,j)\text{ is incorrect}\},
 \qquad (c_y(i,j))_{i\in C_j^A(y)}.
\end{equation}
In particular, $|C_j^A(y)|=d_j^A(y)$, see \eqref{eq:correction-counts}, and the data at a recorded position consist of one of the three stable labels in \eqref{eq:stable-labels}.

The decoder receives the matrix, the parameters, the ordered row set \eqref{eq:row-query-set}, and the record \eqref{eq:correction-record}, but no coefficients of the target. Starting from \eqref{eq:initial-law-step}, it follows the query order \eqref{eq:query-order}. At an unrecorded position it returns the label predicted by \eqref{eq:scores} and leaves the distribution unchanged, while at a recorded position it returns the stored label and performs \eqref{eq:exp-update}. The recorded position determines the row and level, so the current prediction and stored label determine the loss vector \eqref{eq:loss-vector} without any additional real parameter.

\begin{lemma}[Agreement of the two procedures]\label{lem:decoder}
The encoder and decoder have identical predictor states at every query, and their returned labels agree with the target label at every stable comparison in \eqref{eq:stable-labels}.
\end{lemma}
\begin{proof}
The initial distributions agree by \eqref{eq:initial-law-step}. If the states agree before a query, then the two procedures make the same prediction from \eqref{eq:scores}. At a recorded mistake they use the same loss vector and update, see \eqref{eq:loss-vector}--\eqref{eq:exp-update}, while at every unrecorded query both states remain unchanged. Induction therefore proves agreement of the states, and the assertion about stable labels follows because an unrecorded stable prediction was already correct.
\end{proof}

Let the decoded label and the largest level with nonzero decoded label be denoted by
\begin{equation}\label{eq:decoded-level}
 \widehat c_y^A(i,j)\in\{+,-,0\},\qquad
 J_y^A(i):=\max\{j\in\mathcal J:\widehat c_y^A(i,j)\in\{+,-\}\},
\end{equation}
where the maximum is left undefined when the set is empty. Define
\begin{equation}\label{eq:decoded-v}
 v_i^A(y):=
 \begin{cases}
 b_{J_y^A(i)},&J_y^A(i)\text{ is defined},\\
 0,&J_y^A(i)\text{ is undefined},
 \end{cases}
\end{equation}
and the disjoint bands
\begin{equation}\label{eq:bands}
 I_j^A(y):=\{i\in A:J_y^A(i)=j\}.
\end{equation}
The reconstructed energy function is therefore
\begin{equation}\label{eq:v-band-decomp}
 v^A(y)=\sum_{j\in\mathcal J}b_j1_{I_j^A(y)}\quad\text{on }A,
\end{equation}
with each row belonging to at most one band. The largest-level rule in \eqref{eq:decoded-level} is useful twice: it converts the threshold data into a function taking values in the lacunary sequence \eqref{eq:amplitude-levels}, and it ensures that the band \eqref{eq:bands} is determined before lower levels are processed. Monotonicity of the decoded labels across levels is not required, because the relative-error estimate depends only on the largest nonzero level and the next level above it.

\begin{lemma}[Accuracy of the lacunary approximation]\label{lem:quantization}
For the reconstruction \eqref{eq:decoded-v} and row energies \eqref{eq:row-energies},
\begin{equation}\label{eq:pointwise-quantization}
 |v_i^A(y)-z_i(y)|\le C\eta z_i(y)+C\eta^2
 \qquad(i\in A,\ y\in\mathcal K_r).
\end{equation}
In particular, writing
\begin{equation}\label{eq:VA-ZA-definitions}
 V_A(y):=\sum_{i\in A}v_i^A(y),\qquad
 Z_A(y):=\sum_{i\in A}z_i(y),
\end{equation}
one has
\begin{equation}\label{eq:VA-ZA}
 V_A(y)\le(1+C\eta)Z_A(y)+C\eta^2|A|.
\end{equation}
\end{lemma}
\begin{proof}
Fix a row and put $t=|(Hy)_i|$, so $t\le\sqrt r$ by \eqref{eq:basic-energy}. The top query has stable label zero by \eqref{eq:level-endpoints}, and hence the decoder returns zero there by Lemma~\ref{lem:decoder}. If $v_i^A(y)=b_j>0$, see \eqref{eq:decoded-v}, then the nonzero label at level $j$ excludes a stable zero label and therefore $t>(1-\eta)a_j$. The next level exists and has decoded label zero by maximality in \eqref{eq:decoded-level}, so it excludes a stable nonzero label and
\[
 t<(1+\eta)a_{j+1}=(1+\eta)^2a_j.
\]
Consequently,
\[
 (1+\eta)^{-4}t^2<v_i^A(y)<(1-\eta)^{-2}t^2,
\]
and the estimates $1-(1+\eta)^{-4}\le4\eta$ and $(1-\eta)^{-2}-1\le3\eta$, valid in \eqref{eq:eta-range}, prove the relative-error part of \eqref{eq:pointwise-quantization}. If $v_i^A(y)=0$, the decoded label is zero even at the bottom level, so Lemma~\ref{lem:decoder} and \eqref{eq:level-endpoints} imply
\[
 t<(1+\eta)a_{j_-}\le\frac{\eta}{2},
\]
which accounts for the additive error. Summing \eqref{eq:pointwise-quantization} over the row set proves \eqref{eq:VA-ZA} with the definitions in \eqref{eq:VA-ZA-definitions}.
\end{proof}

The function $v^A(y)$, see \eqref{eq:v-band-decomp}, should therefore be viewed as a reconstruction of the lacunary projection of $|Hy|^2$, with the lowest amplitudes suppressed. It need not coincide with the exact nearest projection from \eqref{eq:nearest-projection}, because either answer is allowed inside a transition interval. Nevertheless \eqref{eq:nearest-projection-error} and \eqref{eq:pointwise-quantization} imply
\begin{equation}\label{eq:decoded-nearest-comparison}
 \big|v_i^A(y)-\pi_\eta(|(Hy)_i|)^2\big|
 \le C\eta z_i(y)+C\eta^2.
\end{equation}
This controlled freedom is what permits an economical record: one retains the corrections needed to recover the relevant level sets, rather than the full list of target amplitudes. In \cite{BIK}, the corresponding record recovers approximate coefficient superlevel sets, while here the disjoint bands \eqref{eq:bands} permit the squared amplitudes themselves to be summed against the random selectors.

\section{Probability and counting preliminaries}
\label{sec:probability}

We now collect the finite counting and probability estimates which will enter the proof, keeping the dependence on the accuracy parameter $\eta$ visible until the final absorption argument. The first calculation converts the correction bound into an estimate for the number of possible descriptions, while the second exploits the order in which those descriptions are read. Together, they permit a uniform estimate for random signed sums of the approximations from \eqref{eq:decoded-v}.

\subsection{Counting the correction records}
\label{sec:counting}

Fix an ordered row set $A$ of cardinality $M$, as in \eqref{eq:row-query-set}, and recall that a correction record at one level specifies a collection of rows and one of three labels at each of those rows. Accordingly, define
\begin{equation}\label{eq:sample-code-count}
 N_M(d):=\binom Md3^d,\qquad h_M(d):=\log N_M(d),\qquad 0\le d\le M.
\end{equation}
The elementary estimate $\binom Md\le\left(\frac{eM}{d}\right)^d$, obtained from $d!\ge\left(\frac{d}{e}\right)^d$, implies
\begin{equation}\label{eq:binomial-entropy}
 h_M(d)\le d\log\frac{3eM}{d}\quad(1\le d\le M),\qquad h_M(0)=0.
\end{equation}
We shall estimate a weighted sum of these logarithms, since the amplitude of a band determines how much its description contributes to the signed-sum estimate. The logarithm of the number of complete records would discard this information.

\begin{lemma}[Weighted counting]\label{lem:sample-entropy}
For a target $y\in\mathcal K_r$, write $d_j=d_j^A(y)$ and $D=D_A(y)$, see \eqref{eq:correction-counts}. Then
\begin{equation}\label{eq:weighted-log-sum}
 \sum_{j\in\mathcal J}b_jh_M(d_j)
 \le D\log\frac{3eM\sum_jb_j}{D},
\end{equation}
where both sides are interpreted as zero when $D=0$. In particular, if $M\le4m$, then
\begin{equation}\label{eq:weighted-entropy-final}
 \sum_jb_jh_M(d_j)
 \le\frac{C}{\eta^2}rL_{n,r}
       \log\left(e+\frac{m}{L_{n,r}}\right),
\end{equation}
with $L_{n,r}$ as in \eqref{eq:logarithms}. The same conclusion holds, after changing $C$, under the assumption $M\le8m$.
\end{lemma}

\begin{proof}
When $D>0$, put $x_j=b_jd_j$ for the indices with $d_j>0$, so that \eqref{eq:binomial-entropy} and Jensen's inequality imply
\begin{align*}
 \sum_jb_jh_M(d_j)
 &\le\sum_{j:x_j>0}x_j\log\frac{3eMb_j}{x_j}\\
 &\le D\log\left(\sum_{j:x_j>0}\frac{x_j}{D}
                         \frac{3eMb_j}{x_j}\right)
 \le D\log\frac{3eM\sum_jb_j}{D}.
\end{align*}
This proves \eqref{eq:weighted-log-sum}, with the zero case following from the positivity of the weights.

To use the correction bound without losing monotonicity, introduce
\begin{equation}\label{eq:counting-cutoffs}
 S_b:=\sum_jb_j,\qquad T:=MS_b,\qquad
 D_0:=\frac{32}{3\eta^2}rL_{n,r}.
\end{equation}
Proposition~\ref{prop:mistake-budget} and the inequalities $d_j\le M$ imply $D\le\min(D_0,T)$, while the function 
\[ u\mapsto u\log\frac{3eT}{u} \]
is increasing on $0<u\le T$. If $T\le D_0$, the value of this function at $D$ is at most $T\log(3e)$; if $T>D_0$, its value at $D$ is at most $D_0\log\frac{3eT}{D_0}$. Consequently, in either case,
\[
 D\log\frac{3eT}{D}\le CD_0\log\left(e+\frac{T}{D_0}\right).
\]
The estimate $S_b\le \frac{Cr}{\eta}$, see \eqref{eq:geometric-bounds}, now implies $\frac{T}{D_0}\le \frac{C\eta m}{L_{n,r}}$ when $M\le4m$ or $M\le8m$, which proves \eqref{eq:weighted-entropy-final}. If $M=0$, all the quantities to be estimated vanish.
\end{proof}

The weighted count must now be converted into one event controlling the signed sums for every target, while keeping track of which records determine each band. We therefore exploit the decreasing order of the levels, which allows geometric summation to charge the repeated use of a correction to its original amplitude weight.

\subsection{Random signs and the order of the levels}
\label{sec:signs}

For the counting argument, the decoder is run on every possible correction record, including records which do not arise from a target in $\mathcal K_r$. This extension is well defined because a stored label specifies the update in \eqref{eq:exp-update}, even when that label agrees with the prediction and hence the loss vector is zero. Since the levels are visited in decreasing order, see \eqref{eq:query-order}, the description of a band involves only records which have already been read.

\begin{lemma}[Dependence on higher levels]\label{lem:tail-determinism}
The band $I_j^A(y)$ from \eqref{eq:bands} is determined by the correction records at levels $k\ge j$, so changing any records below $j$ leaves this band unchanged.
\end{lemma}

\begin{proof}
A row belongs to $I_j^A(y)$ precisely when its decoded label is nonzero at level $j$ and zero at every higher level, and all of these labels have been generated before any record below $j$ is consulted.
\end{proof}

For a fixed level $j$, define the corresponding tail profile and family of nonempty bands by
\begin{equation}\label{eq:tail-profile}
 d_{\ge j}:=(d_k)_{k\in\mathcal J,\ k\ge j},\qquad
 \mathcal G^A_{j,d_{\ge j}}
 :=\{I_j\ne\emptyset:I_j\text{ is decoded from records with counts }d_{\ge j}\}.
\end{equation}
Lemma~\ref{lem:tail-determinism} and \eqref{eq:sample-code-count} imply, whenever this family is nonempty,
\begin{equation}\label{eq:tail-family-count}
 \log|\mathcal G^A_{j,d_{\ge j}}|
 \le q_j(d),\qquad q_j(d):=\sum_{k\ge j}h_M(d_k).
\end{equation}
To sum failure probabilities over profiles as well as over bands, put
\begin{equation}\label{eq:Kraft-prior}
 Z_M:=\sum_{d=0}^M\frac{1}{N_M(d)},\qquad
 \pi_M(d):=\frac{1}{Z_MN_M(d)}.
\end{equation}
Since $N_M(d)\ge3^d$, these normalization constants satisfy
\begin{equation}\label{eq:ZM}
 1\le Z_M\le\sum_{d\ge0}3^{-d}=\frac{3}{2},
\end{equation}
and hence, for a prescribed $0<\rho<\frac{1}{4}$,\footnote{We choose $\rho=n^{-(A_0+3)}$ in \eqref{eq:rho-choice}.} we may allocate the failure probabilities by
\begin{equation}\label{eq:failure-weight}
 \pi_{\ge j}(d_{\ge j}):=\prod_{k\ge j}\pi_M(d_k),\qquad
 p_{j,d_{\ge j}}:=\frac{\rho}{L}\pi_{\ge j}(d_{\ge j}).
\end{equation}
Here $L$ is the number of levels, so the allowances in \eqref{eq:failure-weight} sum to $\rho$ over all levels and profiles. These are deterministic weights for a union bound, with every family in \eqref{eq:tail-profile} fixed before the signs are drawn. In particular, no target-dependent probability distribution is sampled.

\begin{proposition}[Uniform control of signed sums]\label{prop:rademacher}
Fix $A\subseteq[n]$, and let $(\sigma_i)_{i\in A}$ be independent random variables taking the values $-1$ and $1$ with equal probability. With probability at least $1-\rho$,
\begin{equation}\label{eq:rademacher-selfnorm}
 \left|\sum_{i\in A}\sigma_iv_i^A(y)\right|
 \le C\sqrt{V_A(y)K_A(y)}\qquad\text{for every }y\in\mathcal K_r,
\end{equation}
where $v^A$ and $V_A$ are defined in \eqref{eq:decoded-v} and \eqref{eq:VA-ZA-definitions}, respectively, and
\begin{equation}\label{eq:KA-def}
 K_A(y):=\sum_{j:I_j^A(y)\ne\emptyset}b_jQ_j^A(y),
\end{equation}
with
\begin{equation}\label{eq:Qj}
 Q_j^A(y):=\log\frac{2|\mathcal G^A_{j,d_{\ge j}(y)}|}{p_{j,d_{\ge j}(y)}}
 \qquad\bigl(I_j^A(y)\ne\emptyset\bigr).
\end{equation}
Moreover, if $M=|A|\le4m$, then uniformly in $y$,
\begin{equation}\label{eq:KA-bound}
 K_A(y)\le C\left[
 \frac{rL_{n,r}}{\eta^3}\log\left(e+\frac{m}{L_{n,r}}\right)
 +\frac{r}{\eta}\log\frac{L}{\rho}+\frac{r}{\eta^2}\right].
\end{equation}
\end{proposition}

\begin{proof}
For a nonempty fixed set $I\subseteq A$, independence and the inequality $\cosh t\le e^{\frac{t^2}{2}}$ imply
\[
 \mathbb E\exp\left(t\sum_{i\in I}\sigma_i\right)
 \le e^{\frac{|I|t^2}{2}},
\]
so Markov's inequality, optimized separately for the two tails, yields
\[
 \mathbb P\left(\left|\sum_{i\in I}\sigma_i\right|>
                     \sqrt{2|I|Q}\right)\le2e^{-Q}.
\]
For a nonempty family $\mathcal G=\mathcal G^A_{j,d_{\ge j}}$, choose $Q=\log\frac{2|\mathcal G|}{p_{j,d_{\ge j}}}$. The union bound over $I\in\mathcal G$ then has cost at most $p_{j,d_{\ge j}}$. Summing these costs by \eqref{eq:failure-weight}, we obtain an event of probability at least $1-\rho$ on which
\begin{equation}\label{eq:band-rademacher}
 \left|\sum_{i\in I_j^A(y)}\sigma_i\right|
 \le C\sqrt{|I_j^A(y)|Q_j^A(y)}
\end{equation}
for every nonempty decoded band, and therefore for every target under consideration. The decomposition in \eqref{eq:v-band-decomp} and Cauchy--Schwarz imply
\begin{align*}
 \left|\sum_{i\in A}\sigma_iv_i^A(y)\right|
 &\le C\sum_{j:I_j^A(y)\ne\emptyset}
                b_j\sqrt{|I_j^A(y)|Q_j^A(y)}\\
 &\le C\left(\sum_jb_j|I_j^A(y)|\right)^{\frac{1}{2}}
        \left(\sum_{j:I_j^A(y)\ne\emptyset}b_jQ_j^A(y)\right)^{\frac{1}{2}},
\end{align*}
which is \eqref{eq:rademacher-selfnorm}. When there are no nonempty bands, both sides vanish, so this also covers $A=\emptyset$.

To bound the cost in \eqref{eq:KA-def}, combine \eqref{eq:tail-family-count}, \eqref{eq:Kraft-prior}, and \eqref{eq:failure-weight} to obtain
\[
 Q_j^A(y)\le C+\log\frac{L}{\rho}
           +2\sum_{k\ge j}h_M(d_k(y))
           +(j_+-j+1)\log Z_M;
\]
the higher-level records contribute
\begin{equation}\label{eq:reversed-record-sum}
 \sum_jb_j\sum_{k\ge j}h_M(d_k)
 =\sum_kh_M(d_k)\sum_{j\le k}b_j
 \le\frac{C}{\eta}\sum_kb_kh_M(d_k),
\end{equation}
where the last step uses geometric summation from \eqref{eq:geometric-bounds}. Lemma~\ref{lem:sample-entropy} controls this term, while \eqref{eq:geometric-bounds}, \eqref{eq:geometric-tail-length}, and \eqref{eq:ZM} control the remaining terms. Together they prove \eqref{eq:KA-bound}.
\end{proof}

The reversal in \eqref{eq:reversed-record-sum} is the probabilistic counterpart of retaining the predictor throughout the levels. A record at level $k$ participates in every band below it, but the total weight of those appearances is at most $\frac{Cb_k}{\eta}$, so the number of levels does not multiply the principal counting cost. This is the same structural feature used in the $\Lambda(p)$ argument of \cite{BIK}: the description available at a given scale includes the corrections made earlier, and the geometric weights permit those cumulative descriptions to be summed against one descent estimate. The potentials differ between the two papers, as explained after \eqref{eq:potential-drop}, while the retention of previous corrections is common to both proofs.

The signed-sum estimate now applies on any fixed row set, and our next step is to connect it with the selectors in \eqref{eq:selectors}. Comparing two independent samples will produce the required independent signs after conditioning on their symmetric difference, while concentration for one target in the second sample will later relate the uniform comparison to the expected energy.

\subsection{Two samples and their symmetric difference}
\label{sec:ghost}

The first sample may fail for a target chosen after its rows are known.
To examine such a target, we introduce a second, independent sample,
usually called a ghost sample. Once the first sample and the target
are fixed, concentration for that one target applies to the second
sample; comparing the two sampled energies will then connect the
uniform signed-sum estimate with the desired bound for the first sample.

So, let $(\delta_i')_{i=1}^n$ be an independent copy of the Bernoulli selectors from \eqref{eq:selectors}, with the same parameter $p=\frac{m}{n}\le\frac{1}{2}$. We shall compare their sampled energies by conditioning on the set of rows selected by exactly one sample, namely
\begin{equation}\label{eq:symmetric-difference}
 A:=\{i:\delta_i\ne\delta_i'\},\qquad
 \sigma_i:=\delta_i-\delta_i'\in\{-1,1\}\quad(i\in A).
\end{equation}
In particular, $A$ is the symmetric difference of the two row sets, and if it is referred to as the active set, this means precisely the set in \eqref{eq:symmetric-difference}. Since the outcomes $(1,0)$ and $(0,1)$ each have probability $p(1-p)$, the signs are independent Rademacher variables conditional on $A$. Ordering $A$ by increasing row index also fixes the order required by the decoder.

The cardinality has the law
\begin{equation}\label{eq:M-law}
 M:=|A|\sim\operatorname{Bin}(n,2p(1-p)),\qquad\mathbb EM\le2m,
\end{equation}
so the bound $\mathbb Ee^{tM}\le\exp(2m(e^t-1))$, followed by Markov's inequality with $t=\log2$, implies
\begin{equation}\label{eq:M-tail}
 \mathbb P(M>4m)\le e^{-(4\log2-2)m}\le e^{-cm}.
\end{equation}
For each $y\in\mathcal K_r$, define
\begin{equation}\label{eq:X-Xprime}
 X_y:=\sum_i\delta_i z_i(y),\qquad X_y':=\sum_i\delta_i'z_i(y),
\end{equation}
where $z_i(y)$ was defined in \eqref{eq:row-energies}. Cancellation on rows selected twice and positivity on the remaining rows imply
\begin{equation}\label{eq:difference-and-mass}
 X_y-X_y'=\sum_{i\in A}\sigma_i z_i(y),\qquad
 Z_A(y)\le X_y+X_y'.
\end{equation}
Here $Z_A$ is the true energy on $A$, see \eqref{eq:VA-ZA-definitions}, while \eqref{eq:rademacher-selfnorm} estimates the corresponding signed sum for its lacunary approximation.

The following one-target estimate will suffice for the independent sample, because uniformity has already been built into Proposition~\ref{prop:rademacher}.

\begin{lemma}[Concentration for one target]\label{lem:ghost-concentration}
For a fixed $y\in\mathcal K_r$, the variable in \eqref{eq:X-Xprime} satisfies
\begin{equation}\label{eq:ghost-variance}
 \mathbb EX_y'=m,\qquad\operatorname{Var}(X_y')\le mr.
\end{equation}
Consequently, if $m\ge\frac{64r}{\varepsilon^2}$, then
\begin{equation}\label{eq:ghost-good-prob}
 \mathbb P\left(|X_y'-m|\le\frac{\varepsilon m}{4}\right)\ge\frac{3}{4}.
\end{equation}
\end{lemma}

\begin{proof}
By \eqref{eq:basic-energy}, the expectation is $p\sum_i z_i(y)=m$, while independence implies
\[
 \operatorname{Var}(X_y')=p(1-p)\sum_i z_i(y)^2
 \le pr\sum_i z_i(y)=mr.
\]
Chebyshev's inequality bounds the complementary probability in \eqref{eq:ghost-good-prob} by $\frac{16r}{\varepsilon^2m}$, which is at most $\frac{1}{4}$ under the stated assumption.
\end{proof}

The comparison with a second sample will involve the sampled energies on both sides of an inequality, so we first isolate the numerical step that will absorb this dependence. The accompanying logarithmic comparison will ensure that the prescribed failure probability fits within the same row count, allowing both estimates to be used directly in the proof of Theorem~\ref{thm:main}.

\subsection{Two elementary numerical estimates}
\begin{lemma}[A comparison of logarithms]\label{lem:log-absorb}
For $1\le r\le n$, the quantities in \eqref{eq:logarithms} satisfy
\begin{equation}\label{eq:logn-product}
 \log n\le CL_{n,r}L_r.
\end{equation}
\end{lemma}

\begin{proof}
Both $a=\log\frac{2en}{r}$ and $b=\log(2r)$ are bounded below by positive absolute constants, so $ab\ge c(a+b)=c\log(4en)$, from which \eqref{eq:logn-product} follows.
\end{proof}

\begin{lemma}[Absorption]\label{lem:absorption}
Fix $0<\varepsilon<\frac{1}{4}$ and an absolute constant $C$. There exist $c_0,c_1>0$, depending only on $C$, with the following property. Suppose that $u,v,m,K\ge0$ satisfy $|v-m|\le\frac{\varepsilon m}{4}$ and
\begin{equation}\label{eq:abstract-selfnorm}
 |u-v|\le C\sqrt{K(u+v+\eta^2m)}+C\eta(u+v)+C\eta^2m.
\end{equation}
If
\begin{equation}\label{eq:absorption-hyp}
 \eta\le c_0\varepsilon,\qquad K\le c_1\varepsilon^2m,
\end{equation}
then
\begin{equation}\label{eq:absorption-conclusion}
 |u-m|\le\varepsilon m.
\end{equation}
\end{lemma}

\begin{proof}
If $m=0$, the assumptions force $v=K=0$, after which \eqref{eq:abstract-selfnorm} implies $u=0$ provided $C\eta<1$. We may therefore assume $m>0$ and put $s=u+v$, so that $u\le v+|u-v|$ and $v\le(1+\frac{\varepsilon}{4})m$ imply
\[
 s\le Cm+C\sqrt{K(s+\eta^2m)}+C\eta s+C\eta^2m.
\]
Choose $\eta$ so that $C\eta\le\frac{1}{8}$, and apply Young's inequality in the forms
\[
 C\sqrt{Ks}\le\frac{s}{4}+C'K,\qquad
 \sqrt{K\eta^2m}\le\frac{\eta}{2}(K+m).
\]
Using \eqref{eq:absorption-hyp} and absorbing the terms proportional to $s$, we conclude that
\begin{equation}\label{eq:s-upper}
 s=u+v\le C_2m.
\end{equation}
Returning to \eqref{eq:abstract-selfnorm} with this preliminary bound, we obtain
\[
 |u-v|\le C_3\sqrt{Km}+C_4\eta m\le\frac{3\varepsilon m}{4}
\]
once $c_0,c_1$ are sufficiently small, and the assumption $|v-m|\le\frac{\varepsilon m}{4}$ then proves \eqref{eq:absorption-conclusion}.
\end{proof}
\medskip
With these preliminaries in mind, we proceed to the main body of our proofs.
\section{Proof of the sampling and recovery theorems}
\label{sec:completion}

We first prove the real matrix theorem by comparing the two samples from \eqref{eq:X-Xprime}, after which the Fourier-ratio statements follow from the complex extension and an elementary minimization argument. The additional point in the complex extension is the dependence between the real and imaginary components of one sampled row, which is treated in detail in \S\ref{sec:complex}.

\subsection{The difference between two sampled energies}

Fix the failure allowance
\begin{equation}\label{eq:rho-choice}
 \rho:=n^{-(A_0+3)}, 
\end{equation}
and retain the accuracy parameter $\eta$ from \eqref{eq:eta-range}. The next proposition combines the lacunary approximation with the signed-sum estimate, producing a bound in terms of the two sampled energies themselves.

\begin{proposition}[Comparison of two samples]\label{prop:ghost-difference}
There is an event $\mathcal G$, depending on $(\delta,\delta')$, such that
\begin{equation}\label{eq:G-prob}
 \mathbb P(\mathcal G^c)\le\rho+e^{-cm},
\end{equation}
and on this event, simultaneously for every $y\in\mathcal K_r$,
\begin{equation}\label{eq:selfnorm-ghost}
 |X_y-X_y'|\le
 C\sqrt{K_m(X_y+X_y'+\eta^2m)}
 +C\eta(X_y+X_y')+C\eta^2m,
\end{equation}
where $X_y,X_y'$ are defined in \eqref{eq:X-Xprime}, and
\begin{equation}\label{eq:Km}
 K_m:=C_{\eta,A_0}\left[
 rL_{n,r}\log\left(e+\frac{m}{L_{n,r}}\right)+r\log n\right].
\end{equation}
\end{proposition}

\begin{proof}
Condition on a symmetric difference $A$ from \eqref{eq:symmetric-difference} with $M=|A|\le4m$, so that Proposition~\ref{prop:rademacher} applies to its independent signs. Since $L\le C_\eta\log(2r)$ by \eqref{eq:geometric-bounds}, the choice \eqref{eq:rho-choice} makes the right side of \eqref{eq:KA-bound} at most the deterministic quantity $K_m$ from \eqref{eq:Km}, after increasing its constant. On the event in \eqref{eq:rademacher-selfnorm}, Lemma~\ref{lem:quantization} therefore implies
\begin{align*}
 |X_y-X_y'|
 &\le\left|\sum_{i\in A}\sigma_iv_i^A(y)\right|
       +C\eta Z_A(y)+C\eta^2M\\
 &\le C\sqrt{K_mV_A(y)}+C\eta Z_A(y)+C\eta^2M.
\end{align*}
Here the decoded energy $V_A$ and the true energy $Z_A$ are as in \eqref{eq:VA-ZA-definitions}, while \eqref{eq:VA-ZA} bounds the former by $C(Z_A+\eta^2M)$. Substituting $M\le4m$ and $Z_A\le X_y+X_y'$, see \eqref{eq:difference-and-mass}, proves \eqref{eq:selfnorm-ghost}.

For every fixed $A$ with $|A|\le4m$, the conditional failure probability is at most $\rho$, so averaging over those sets incurs total cost at most $\rho$. The remaining event has probability at most $e^{-cm}$ by \eqref{eq:M-tail}, proving \eqref{eq:G-prob}. In particular, the conditioning involves no union bound over the possible symmetric differences.
\end{proof}

The logarithm in \eqref{eq:Km} contains the expected number of sampled rows, because the correction positions were counted within $A$. The next calculation shows that this dependence is consistent with the row threshold in \eqref{eq:m-range}, once $\eta$ is chosen as a sufficiently small multiple of $\varepsilon$.

\begin{lemma}[The required size of the sample]\label{lem:Km-threshold}
Fix $0<\varepsilon<\frac{1}{4}$ and $A_0>0$, and choose $\eta=\frac{c_0\varepsilon}{2}$ with $c_0$ as in Lemma~\ref{lem:absorption}. If
\begin{equation}\label{eq:threshold-again}
 m\ge C_{\varepsilon,A_0}rL_{n,r}L_r,
\end{equation}
with a sufficiently large constant, then
\begin{equation}\label{eq:Km-small}
 K_m\le c_1\varepsilon^2m,
\end{equation}
where $K_m$ is defined in \eqref{eq:Km} and $c_1$ is the constant in \eqref{eq:absorption-hyp}.
\end{lemma}

\begin{proof}
Lemma~\ref{lem:log-absorb} bounds $r\log n$ by $CrL_{n,r}L_r$, while for the first term in \eqref{eq:Km} we introduce
\begin{equation}\label{eq:dimensionless-sample-count}
 x:=\frac{m}{rL_{n,r}}.
\end{equation}
The hypothesis implies $x\ge C_{\varepsilon,A_0}L_r$, and
\[
 \log\left(e+\frac{m}{L_{n,r}}\right)
 =\log(e+rx)\le CL_r+\log(1+x).
\]
Given $\gamma>0$, choose $T$ sufficiently large that $\frac{CL_r}{x}\le\frac{\gamma}{2}$ for $x\ge TL_r$, and that $\frac{\log(1+x)}{x}\le\frac{\gamma}{2}$ for $x\ge T\log2$. This is possible because the latter quotient decreases to zero. Taking $\gamma$ small relative to $\frac{c_1\varepsilon^2}{C_{\eta,A_0}}$, and then increasing the constant in \eqref{eq:threshold-again} to control both terms, proves \eqref{eq:Km-small}.
\end{proof}
With these preliminaries in mind, we proceed to the main body of the proof.

\begin{proof}[Proof of Theorem~\ref{thm:main}]
Choose $\eta=\frac{c_0\varepsilon}{2}$, reducing $c_0$ if necessary
so that \eqref{eq:eta-range} holds, and take the constant in
\eqref{eq:m-range} large enough for
Lemmas~\ref{lem:ghost-concentration} and \ref{lem:Km-threshold}.
We now make explicit the dependence on the sampling
configuration that was suppressed in \eqref{eq:X-Xprime},
writing
\begin{equation}\label{eq:sampled-energy-explicit}
 X(\delta,y):=\sum_{i=1}^n\delta_i z_i(y),
 \qquad
 \delta=(\delta_i)_{i=1}^n\in\{0,1\}^n,\quad
 y\in\mathcal K_r.
\end{equation}
Here $\delta$ specifies the selected rows, while $y$ specifies
the target. The sampling law is fixed by \eqref{eq:selectors},
and the same realization $\delta$ is used as $y$ ranges over
$\mathcal K_r$. Define the failure set by
\begin{equation}\label{eq:bad-sample-event}
 \mathcal B:=
 \left\{\delta\in\{0,1\}^n:
 \sup_{y\in\mathcal K_r}|X(\delta,y)-m|>\varepsilon m
 \right\};
\end{equation}
in particular, $\delta\in\mathcal B$ precisely when the selected rows
fail the desired estimate for at least one target, and
$\mathbb P(\mathcal B)$ denotes the probability of this
set under the Bernoulli sampling law.

For each fixed $\delta$, the map
$y\mapsto |X(\delta,y)-m|$ is continuous on the compact
class $\mathcal K_r$, see \eqref{eq:target-class}, so its
maximum is attained. For every $\delta\in\mathcal B$,
choose a maximizing target $y(\delta)$, which consequently
satisfies
\begin{equation}\label{eq:bad-sample-witness}
 |X(\delta,y(\delta))-m|>\varepsilon m;
\end{equation}
on the complement of $\mathcal B$, set $y(\delta)=e_1$,
so that the target is defined for every sampling
configuration. This choice is measurable because the
configuration space is finite.

The witness depends only on the first sample, so after
$\delta$ is fixed, $y(\delta)$ is fixed as well and the
second selector vector $\delta'$ retains its original
independent Bernoulli law. In the notation used below,
\[
 X_{y(\delta)}=X(\delta,y(\delta)),
 \qquad
 X_{y(\delta)}'=X(\delta',y(\delta)),
\]
where the first quantity is determined by the fixed
sample and its witness, while the second remains random
through $\delta'$.

Fix $\delta\in\mathcal B$ and its witness, so that $|X_{y(\delta)}-m|>\varepsilon m$, while the independent sample remains random. By \eqref{eq:ghost-good-prob}, with conditional probability at least $\frac{3}{4}$ it satisfies
\[
 |X_{y(\delta)}'-m|\le\frac{\varepsilon m}{4}.
\]
If $\mathcal G$ also holds, \eqref{eq:selfnorm-ghost} and \eqref{eq:Km-small} meet the hypotheses of Lemma~\ref{lem:absorption} with $u=X_{y(\delta)}$ and $v=X_{y(\delta)}'$, forcing $|X_{y(\delta)}-m|\le\varepsilon m$. This contradicts the choice of witness, and hence
\begin{equation}\label{eq:bad-event-comparison}
 \frac{3}{4}\mathbb P(\mathcal B)
 \le\mathbb P\left(\delta\in\mathcal B, \text{ and }             |X_{y(\delta)}'-m|\le\frac{\varepsilon m}{4}\right)
 \le\mathbb P(\mathcal G^c).
\end{equation}
By \eqref{eq:G-prob}, the last probability is at most $\rho+e^{-cm}$, and Lemma~\ref{lem:log-absorb} permits this to be bounded by $Cn^{-A_0}$ after increasing the constant in \eqref{eq:m-range}. We have proved \eqref{eq:bernoulli-rip} uniformly over the entire norm class. Applying that estimate to $\frac{y}{\|y\|_2}$ proves \eqref{eq:homogeneous-sampling} for nonzero $y$, with the zero vector satisfying it trivially. Since every $r$-sparse vector obeys $\|y\|_1\le\sqrt r\|y\|_2$, the RIP assertion follows from \eqref{eq:rip-definition}.
\end{proof}

With the matrix estimate established in the real setting, we return to the original recovery problem, using the complex and prescribed-cardinality extensions proved in the complements. The lower energy estimate combines with minimality in \eqref{eq:fourier-minimization} to control the reconstruction error, which explains why uniform sampling is the input needed for Corollary~\ref{cor:fourier-ratio-accuracy}.

\subsection{Fourier-ratio sampling and recovery}
\label{subsec:fourier-ratio}

We now return to the Fourier-analytic formulation from the introduction, using the complex matrix theorem whose remaining details are established in \S\ref{sec:complex}--\S\ref{sec:prescribed}. The order of exposition separates the scalar approximation and its probability estimates from the treatment of complex-valued measurements.

\begin{proof}[Proof of Theorem~\ref{thm:fourier-ratio-sampling}]
Let
\begin{equation}\label{eq:inverse-fourier-matrix}
 F:=\bigl(e^{\frac{2\pi ix\xi}{n}}\bigr)_{x,\xi\in\mathbb Z_n},\qquad y:=\widehat f,
\end{equation}
so that $F^*F=nI$ and Fourier inversion in \eqref{eq:unitary-fourier} implies $Fy=\sqrt n f$. The norm condition defining $\mathcal F_r$, see \eqref{eq:fourier-class}, becomes $\|y\|_1\le\sqrt r\|y\|_2$, while Plancherel identifies $\|y\|_2$ with $\|f\|_2$. Theorem~\ref{thm:complex-main}, with the prescribed-cardinality conclusion proved in \S\ref{sec:prescribed}, therefore implies \eqref{eq:fourier-ratio-energy} on one event for every such function.
\end{proof}

The lower estimate in \eqref{eq:fourier-ratio-energy} has the uncertainty consequence
\begin{equation}\label{eq:fourier-ratio-kernel}
 h\ne0,\quad h|_S=0
 \quad\Longrightarrow\quad
 \operatorname{FR}(h)>\sqrt r,
 \quad\text{equivalently}\quad
 \|h\|_2<\frac{\|\widehat h\|_1}{\sqrt r}.
\end{equation}
Indeed, a nonzero function with ratio at most $\sqrt r$ would contradict the lower energy estimate, so the space of functions invisible to the observations obeys a quantitative comparison of coefficient norms. This is the uncertainty-to-recovery mechanism emphasized in \cite{IosevichMayeli,IKLM,BINFourierRatio}, and the next proposition records the deterministic deduction with its noise dependence.

\begin{proposition}[Recovery by Fourier minimization]\label{prop:fourier-ratio-recovery}
Let $S\subseteq\mathbb Z_n$, $m>0$, $1\le r\le n$, and $0\le\varepsilon<1$, and suppose that the map $T_S$ from \eqref{eq:normalized-observations} satisfies
\begin{equation}\label{eq:deterministic-fourier-lower}
 \|T_Sh\|_2^2\ge(1-\varepsilon)\|h\|_2^2
 \qquad(h\in\mathcal F_r).
\end{equation}
For data $b=T_Sf+e$ with $\|e\|_2\le\nu$, the minimization problem in \eqref{eq:fourier-minimization} has a solution, and every minimizer $f^\sharp$ satisfies
\begin{equation}\label{eq:fourier-ratio-noisy-recovery}
 \|f^\sharp-f\|_2
 \le\max\left\{\frac{2\|\widehat f\|_1}{\sqrt r},
                    \frac{2\nu}{\sqrt{1-\varepsilon}}\right\}.
\end{equation}
In particular, for $f\ne0$ and noiseless data,
\begin{equation}\label{eq:fourier-ratio-noiseless-recovery}
 \|f^\sharp-f\|_2
 \le\frac{2\operatorname{FR}(f)}{\sqrt r}\|f\|_2,
\end{equation}
where the ratio is defined in \eqref{eq:fourier-ratio}.
\end{proposition}

\begin{proof}
The true function is feasible, so we may restrict the minimization to feasible $g$ with $\|\widehat g\|_1\le\|\widehat f\|_1$. This is a nonempty closed and bounded set in finite dimensions, since $\|g\|_2=\|\widehat g\|_2\le\|\widehat g\|_1$, and compactness therefore ensures a minimizer. For its error $h=f^\sharp-f$, minimality and feasibility imply
\begin{equation}\label{eq:fourier-recovery-two-bounds}
 \|\widehat h\|_1\le2\|\widehat f\|_1,
 \qquad\|T_Sh\|_2\le2\nu.
\end{equation}
If $h\in\mathcal F_r$, then \eqref{eq:deterministic-fourier-lower} and \eqref{eq:fourier-recovery-two-bounds} yield $\|h\|_2\le\frac{2\nu}{\sqrt{1-\varepsilon}}$. Otherwise, the definition in \eqref{eq:fourier-class} implies $\|h\|_2<\frac{\|\widehat h\|_1}{\sqrt r}\le\frac{2\|\widehat f\|_1}{\sqrt r}$, which proves \eqref{eq:fourier-ratio-noisy-recovery} in the remaining case. Substituting $\nu=0$ and using \eqref{eq:fourier-ratio} proves \eqref{eq:fourier-ratio-noiseless-recovery}. And, if $f=0$, minimality forces $f^\sharp=0$ as well.
\end{proof}

\begin{proof}[Proof of Corollary~\ref{cor:fourier-ratio-accuracy}]
Apply Theorem~\ref{thm:fourier-ratio-sampling} with $\varepsilon=\frac{1}{8}$ and $r$ from \eqref{eq:fourier-ratio-accuracy-parameter}, so that Proposition~\ref{prop:fourier-ratio-recovery} holds simultaneously for every function and admissible error vector. For the functions with $\operatorname{FR}(f)\le R$, the inequality $\frac{2\operatorname{FR}(f)}{\sqrt r}\le\tau$ turns \eqref{eq:fourier-ratio-noisy-recovery} into \eqref{eq:fourier-ratio-relative-recovery}, as claimed.
\end{proof}

The distinction between sampling distortion and reconstruction accuracy is useful here: the former was fixed at $\frac{1}{8}$, whereas the latter was specified by increasing the ratio threshold through \eqref{eq:fourier-ratio-accuracy-parameter}. In particular, the improvement in \eqref{eq:old-new-ratio-counts} does not incur an additional power of $\tau^{-1}$ from the variable-distortion bound in \eqref{eq:explicit-row-count}. If the threshold leaves no admissible $m\le \frac{n}{2}$, one may instead observe all $n$ values, obtaining exact noiseless reconstruction.

For observations written without the normalization in \eqref{eq:normalized-observations}, define
\begin{equation}\label{eq:unnormalized-noise}
 b_0=(f(x))_{x\in S}+e_0,\qquad\|e_0\|_2\le\sigma,\qquad
 b=\sqrt{\frac{n}{m}}b_0,\qquad\nu=\sqrt{\frac{n}{m}}\sigma.
\end{equation}
Then the constraint in \eqref{eq:fourier-minimization} becomes
\[
 \|(g(x))_{x\in S}-b_0\|_2\le\sigma,
\]
while the noise term in \eqref{eq:fourier-ratio-noisy-recovery} is $2\sqrt{\frac{n}{m(1-\varepsilon)}}\sigma$. This factor accounts for the passage from error on the observed coordinates to the counting-measure norm of the whole signal.

The recovery argument shows how the Fourier ratio controls the final error, while the descent proof uses a distribution representing the coefficient vector. We next make the connection between these two appearances explicit, identifying the coefficient concentration estimate that supplies the initial entropy bound and thereby controls the complexity of the lacunary approximation.

\subsection{The Fourier ratio in the entropy estimate}

The same parameter which describes the signal class also enters the descent potential, a point which can be made explicit by setting, for $f\ne0$,
\begin{equation}\label{eq:fourier-coefficient-masses}
 \alpha_\xi:=\frac{|\widehat f(\xi)|}{\|\widehat f\|_1},
 \qquad \sum_\xi\alpha_\xi=1.
\end{equation}
The collision sum of this probability distribution is
\begin{equation}\label{eq:fourier-ratio-concentration}
 \sum_\xi\alpha_\xi^2
 =\frac{\|\widehat f\|_2^2}{\|\widehat f\|_1^2}
 =\frac{1}{\operatorname{FR}(f)^2},
\end{equation}
and Jensen's inequality, applied with the nonzero weights $\alpha_\xi$, implies
\begin{equation}\label{eq:fourier-ratio-entropy}
 -\sum_\xi\alpha_\xi\log\alpha_\xi
 \ge-\log\sum_\xi\alpha_\xi^2
 =\log\operatorname{FR}(f)^2,
\end{equation}
with the convention $0\log0=0$. Equality holds for equally sized nonzero Fourier coefficients, while \eqref{eq:fourier-ratio-concentration} continues to measure concentration when the coefficients have unequal sizes or full support.

Lemma~\ref{lem:norm-comparator} uses precisely this norm--entropy inequality for the absolute coordinates of a normalized real vector, followed by a mixture with the uniform distribution, see \eqref{eq:target-representation}. The complex construction in \S\ref{sec:complex} retains this bound up to absolute constants, so the Fourier ratio enters before the counting and probability arguments begin. This identifies the analytic input required by the sampling problem: the descent mechanism from \cite{BIK}, in its entropic form developed here, controls the complexity of a lacunary approximation in terms of the ratio that the Fourier program seeks to understand.

\section{Complements}
\label{sec:complements}

The principal argument is complete in the real setting, and we now supply the complex extension used in the Fourier formulation, followed by refinements concerning sample size, distortion, and recovery. As in the complements to the $\Lambda(p)$ argument in \cite{BIK}, the purpose is to identify which parts of the construction survive a change of setting and which require an additional estimate.

\subsection{Complex transforms and paired measurements}
\label{sec:complex}

For a complex matrix satisfying \eqref{eq:complex-matrix-assumptions}, the deterministic predictor can be applied to the real and imaginary components of each measurement. These components are sampled together, however, so the random sign attached to one original row must be shared by its two real queries. We first quantify that dependence, after which the entropy and counting estimates retain their previous form.

\begin{lemma}[A signed sum with paired components]\label{lem:paired-signs}
Let $A$ be finite and $(\sigma_i)_{i\in A}$ be independent Rademacher signs, and for $I\subseteq A\times\{1,2\}$ define
\begin{equation}\label{eq:paired-multiplicities}
 q_i:=1_I(i,1)+1_I(i,2)\in\{0,1,2\}.
\end{equation}
Then, for every real $t$,
\begin{equation}\label{eq:paired-mgf}
 \mathbb E\exp\left(t\sum_{i\in A}\sigma_iq_i\right)
 \le\exp(t^2|I|),
\end{equation}
and consequently, whenever $I\ne\emptyset$ and $Q>0$,
\begin{equation}\label{eq:paired-tail}
 \mathbb P\left(\left|\sum_{i\in A}\sigma_iq_i\right|
                 >2\sqrt{|I|Q}\right)\le2e^{-Q}.
\end{equation}
\end{lemma}

\begin{proof}
The multiplicities in \eqref{eq:paired-multiplicities} satisfy $q_i^2\le2q_i$, so independence implies
\[
 \mathbb E\exp\left(t\sum_i\sigma_iq_i\right)
 =\prod_i\cosh(tq_i)
 \le\exp\left(\frac{t^2}{2}\sum_iq_i^2\right)
 \le\exp(t^2|I|).
\]
Markov's inequality with $t=\frac{u}{2|I|}$ bounds the upper tail at $u$ by $e^{-\frac{u^2}{4|I|}}$, and the same calculation for the negative sum, with $u=2\sqrt{|I|Q}$, proves \eqref{eq:paired-tail}.
\end{proof}

\begin{proof}[Proof of the complex Bernoulli assertion in Theorem~\ref{thm:complex-main}]
Write $H=H_1+iH_2$ and $y=u+iv$ with real matrices and vectors, and introduce the realification
\begin{equation}\label{eq:realified-objects}
 \widetilde H:=\begin{pmatrix}H_1&-H_2\\H_2&H_1\end{pmatrix},
 \qquad \widetilde y:=\begin{pmatrix}u\\v\end{pmatrix}.
\end{equation}
Then
\begin{equation}\label{eq:realification}
 \widetilde H\widetilde y=
 \begin{pmatrix}\operatorname{Re}(Hy)\\\operatorname{Im}(Hy)\end{pmatrix},
 \qquad\|\widetilde y\|_2=\|y\|_2,
 \qquad\|\widetilde y\|_1\le\sqrt2\|y\|_1,
\end{equation}
where the last inequality follows coordinatewise from $|a|+|b|\le\sqrt2|a+ib|$. In particular, a target in $\mathcal K_r(\mathbb C)$, see \eqref{eq:target-class}, produces a real unit vector satisfying the $\ell^1$ constraint with parameter $2r$, while every entry of $\widetilde H$ has absolute value at most one.

The deterministic arguments in \S\ref{sec:predictor} and the record count in \S\ref{sec:counting} require only these entry and norm bounds, so they apply with dimension $2n$ and parameter $2r$. The orthogonality identity is $\widetilde H^T\widetilde H=nI_{2n}$, rather than $2nI_{2n}$, which is why we apply those deterministic arguments directly and return to the original complex energies before using the sampling normalization.

For a fixed row set $A$, use the ordered query set and atoms
\begin{equation}\label{eq:complex-query-data}
 \widetilde A:=A\times\{1,2\},\qquad
 \widetilde{\mathcal V}:=\{\pm\sqrt{2r}e_k:1\le k\le2n\},
 \qquad\widetilde\lambda:=\frac{\eta}{16r}.
\end{equation}
Here the query $(i,1)$ concerns the real component of $(Hy)_i$ and $(i,2)$ concerns its imaginary component, with any fixed ordering of these queries used throughout. Replace $r$ by $2r$ in the level endpoints from \eqref{eq:level-endpoints}, and denote the resulting number of levels by $\widetilde L$. Within this proof, $a_j,b_j$ refer to that enlarged grid. Since
\[
 L_{2n,2r}=\log\frac{2en}{r}=L_{n,r},
\]
Lemma~\ref{lem:norm-comparator} bounds the initial relative entropy by $L_{n,r}$, and Proposition~\ref{prop:mistake-budget} yields
\begin{equation}\label{eq:complex-budget}
 \sum_jb_jd_j^{\widetilde A}(\widetilde y)
 \le\frac{64}{3\eta^2}rL_{n,r}.
\end{equation}
In particular, splitting a complex vector into real components changes the correction bound only by an absolute factor, even when the original vector has full support.

For the corresponding component energies and decoded values, write
\begin{equation}\label{eq:complex-component-energies}
 z_{i,1}(y):=(\operatorname{Re}(Hy)_i)^2,\qquad
 z_{i,2}(y):=(\operatorname{Im}(Hy)_i)^2,\qquad
 v_{i,c}:=v_{(i,c)}^{\widetilde A}(\widetilde y)\quad(c=1,2).
\end{equation}
Lemma~\ref{lem:quantization}, applied to both components, implies
\begin{equation}\label{eq:complex-quantization}
 |v_{i,1}+v_{i,2}-|(Hy)_i|^2|
 \le C\eta|(Hy)_i|^2+C\eta^2.
\end{equation}
The individual $v_{i,c}$ lie in the lacunary grid, see \eqref{eq:decoded-v}, although their sum need not lie in that same grid. This causes no difficulty because the band decomposition is performed on the component query set $\widetilde A$. Define
\begin{equation}\label{eq:complex-total-energies}
 V:=\sum_{i\in A}(v_{i,1}+v_{i,2}),\qquad
 Z_A(y):=\sum_{i\in A}|(Hy)_i|^2,
\end{equation}
so that \eqref{eq:complex-quantization} implies $V\le(1+C\eta)Z_A(y)+C\eta^2|A|$.

The count for one level now uses $N_{2M}(d)=\binom{2M}{d}3^d$, see \eqref{eq:sample-code-count}, because there are $2M$ component queries when $M=|A|$. Whenever $M\le4m$, the last assertion of Lemma~\ref{lem:sample-entropy}, together with \eqref{eq:complex-budget}, implies the weighted count
\begin{equation}\label{eq:complex-weighted-count}
 \sum_jb_jh_{2M}(d_j)
 \le\frac{C}{\eta^2}rL_{n,r}
              \log\left(e+\frac{m}{L_{n,r}}\right).
\end{equation}
The tail families and probability allocations are defined as in \eqref{eq:tail-profile} and \eqref{eq:failure-weight}, with $2M$ and $\widetilde L$ in place of $M$ and $L$. All records are fixed before the signs are exposed, while the signed sum on a band $I_j\subseteq\widetilde A$ is
\[
 \sum_{i\in A}\sigma_i\bigl(1_{I_j}(i,1)+1_{I_j}(i,2)\bigr).
\]
Lemma~\ref{lem:paired-signs} therefore supplies the required tail bound using a single sign per original row, and the same allocation of failure probabilities as before produces an event of conditional probability at least $1-\rho$ on which
\begin{equation}\label{eq:complex-signed}
 \left|\sum_{i\in A}\sigma_i(v_{i,1}+v_{i,2})\right|
 \le C\sqrt{VK}
 \qquad\text{for every }y\in\mathcal K_r(\mathbb C),
\end{equation}
where
\begin{equation}\label{eq:complex-K}
 K:=\sum_{j:I_j\ne\emptyset}b_jQ_j
 \le C\left[
 \frac{rL_{n,r}}{\eta^3}\log\left(e+\frac{m}{L_{n,r}}\right)
 +\frac{r}{\eta}\log\frac{\widetilde L}{\rho}+\frac{r}{\eta^2}\right].
\end{equation}
For clarity, the Cauchy--Schwarz calculation underlying \eqref{eq:complex-signed} is
\[
 \sum_jb_j\sqrt{|I_j|Q_j}
 \le\left(\sum_jb_j|I_j|\right)^{\frac{1}{2}}
     \left(\sum_jb_jQ_j\right)^{\frac{1}{2}}
 =\sqrt{VK},
\]
and the reversal in \eqref{eq:reversed-record-sum} bounds the higher-level contribution to $K$ by $\frac{C}{\eta}$ times \eqref{eq:complex-weighted-count}. The geometric estimates in \eqref{eq:geometric-bounds} and \eqref{eq:geometric-tail-length}, with $2r$ replacing $r$, control its remaining terms and prove \eqref{eq:complex-K}.

We now return to the original complex energies from \eqref{eq:row-energies}, for which orthogonality and the coefficient norm bound imply
\begin{equation}\label{eq:complex-energy}
 \sum_i z_i(y)=n,\qquad z_i(y)\le\|y\|_1^2\le r
 \qquad(y\in\mathcal K_r(\mathbb C)).
\end{equation}
Choose the two samples on the original $n$ rows, with their symmetric difference defined by \eqref{eq:symmetric-difference}. The bound \eqref{eq:M-tail} and the conditional independent-sign description apply to these rows. Combining \eqref{eq:complex-quantization} with \eqref{eq:complex-signed}, and choosing $\rho$ as in \eqref{eq:rho-choice}, proves \eqref{eq:selfnorm-ghost} with adjusted absolute constants. For fixed $\eta$, the cost in \eqref{eq:complex-K} is bounded by \eqref{eq:Km}, since $\widetilde L\le C\eta^{-1}\log\frac{4r}{\eta}$.

Finally, \eqref{eq:complex-energy} implies $\mathbb EX_y'=m$ and $\operatorname{Var}(X_y')\le mr$, so the one-target estimate, absorption, and threshold calculation apply with the original parameters $n,r,m$. The target class is compact as a subset of $\mathbb R^{2n}$, which permits the same witness selection used in \eqref{eq:bad-sample-event}--\eqref{eq:bad-event-comparison}. This proves the complex Bernoulli assertion, and scaling the unit vectors proves \eqref{eq:homogeneous-sampling} for every complex vector satisfying its norm condition.
\end{proof}

\subsection{Sampling a prescribed number of rows}
\label{sec:prescribed}

The Bernoulli model was used to make the selectors independent, while a prescribed sample size follows by conditioning on cardinality. The polynomial failure bound absorbs the cost of this conditioning, with no change to the simultaneous control of the target class.

\begin{corollary}[Uniform sampling without replacement]\label{cor:uniform}
Fix $0<\varepsilon<\frac{1}{4}$ and $A_0>0$, and let $H$ satisfy \eqref{eq:complex-matrix-assumptions}. If $m$ is an integer with
\begin{equation}\label{eq:uniform-row-count}
 C_{\varepsilon,A_0}rL_{n,r}L_r\le m\le\frac{n}{2},
\end{equation}
then a uniformly chosen $m$-element set $S\subseteq[n]$ satisfies \eqref{eq:homogeneous-sampling} for every complex $y$ obeying $\|y\|_1\le\sqrt r\|y\|_2$, with probability at least $1-Cn^{-A_0}$. In particular, this proves the prescribed-cardinality assertion of Theorem~\ref{thm:complex-main}, including its real specialization.
\end{corollary}

\begin{proof}
Take $p=\frac{m}{n}$ in \eqref{eq:selectors}, so that conditional on cardinality $m$ the row set has the required uniform distribution. The two-sided Stirling estimates imply
\begin{equation}\label{eq:binomial-conditioning-cost}
 \begin{aligned}
 \mathbb P(\operatorname{Bin}(n,p)=m)
 &=\frac{n!}{m!(n-m)!}
   \left(\frac{m}{n}\right)^m
   \left(\frac{n-m}{n}\right)^{n-m}\\
 &\ge c\sqrt{\frac{n}{m(n-m)}}\ge cm^{-\frac{1}{2}},
 \end{aligned}
\end{equation}
where the powers from the factorials cancel against the binomial probability factors. In particular, conditioning multiplies the failure probability by at most $C\sqrt m$, and applying the Bernoulli result with exponent $A_0+1$ proves the stated bound because $\sqrt m\le\sqrt n$. The event before conditioning already controls every target simultaneously, so the same uniformity survives conditioning and scaling.
\end{proof}

For the inverse Fourier matrix in \eqref{eq:inverse-fourier-matrix}, orthogonality follows directly by summing a nontrivial character of $\mathbb Z_n$, and the homogeneous sampling estimate takes the concrete form
\begin{equation}\label{eq:fourier-matrix-sampling}
 (1-\varepsilon)\|y\|_2^2
 \le\frac{1}{m}\sum_{x\in S}\left|\sum_{\xi\in\mathbb Z_n}
                       y_\xi e^{\frac{2\pi ix\xi}{n}}\right|^2
 \le(1+\varepsilon)\|y\|_2^2.
\end{equation}
Substituting $y=\widehat f$ recovers the function formulation \eqref{eq:fourier-ratio-energy}, while arbitrary $y$ in the norm class includes the usual sparse-vector interpretation of Fourier measurements.

\subsection{Dependence on the distortion}
\label{sec:distortion}

We next retain the dependence on $\eta$ in the estimates already proved, which identifies the sources of the sufficient factor $\varepsilon^{-5}$ in \eqref{eq:explicit-row-count}. The correction bound costs $\eta^{-2}$, see \eqref{eq:mistake-budget}, and the reversal of higher-level records costs one further power of $\eta^{-1}$, see \eqref{eq:reversed-record-sum}. Absorption then requires the resulting quantity to be at most $c\varepsilon^2m$, with $\eta$ proportional to $\varepsilon$.

\begin{proposition}[An explicit sufficient row count]\label{prop:explicit-distortion}
For each $A_0>0$, there is a constant $C_{A_0}$ independent of $\varepsilon,n,r$ such that \eqref{eq:explicit-row-count} implies the Bernoulli and prescribed-cardinality conclusions of Theorems~\ref{thm:main} and \ref{thm:complex-main}.
\end{proposition}

\begin{proof}
Use the larger absolute constants required by the complex construction, writing $L$ for the number of levels in either construction, so that $L\le C\eta^{-1}\log\frac{4r}{\eta}$. The bounds \eqref{eq:KA-bound} and \eqref{eq:complex-K} imply that the comparison \eqref{eq:selfnorm-ghost} holds outside an event of probability at most $\rho+e^{-cm}$ with the deterministic choice
\begin{equation}\label{eq:Km-explicit}
 \begin{split}
 a&:=L_{n,r},\qquad\rho:=n^{-(A_0+3)},\\
 K_m&:=C\left[
 \frac{ra}{\eta^3}\log\left(e+\frac{m}{a}\right)
 +\frac{r}{\eta}\log\frac{L}{\rho}+\frac{r}{\eta^2}\right],
 \end{split}
\end{equation}
where all constants are absolute. The approximation and comparison estimates also have absolute constants throughout \eqref{eq:eta-range}, so choose $\eta=c\varepsilon$ with $c>0$ small enough for Lemma~\ref{lem:absorption}, and define
\begin{equation}\label{eq:distortion-parameters}
 U:=\log\frac{2r}{\varepsilon},\qquad x:=\frac{m}{ra}.
\end{equation}
Since $a,U\ge1$ and \eqref{eq:logn-product} implies $\log n\le CaU$, the level count satisfies
\begin{equation}\label{eq:level-probability-cost}
 \log\frac{L}{\rho}\le C_{A_0}aU.
\end{equation}
Indeed, $\log L\le C+\log\frac{1}{\eta}+\log\log\frac{4r}{\eta}\le CU$, while the remaining term is $(A_0+3)\log n$.

Suppose now that
\begin{equation}\label{eq:distortion-threshold}
 x\ge x_0:=B\varepsilon^{-5}U,
\end{equation}
where $B\ge2$ will depend only on $A_0$. The quotient $\frac{\log(e+rx)}{x}$ decreases for $x>0$, because its derivative has the sign of $\frac{rx}{e+rx}-\log(e+rx)<0$, and
\[
 \log(e+rx_0)\le C(1+\log B)U,
\]
since $\log r$, $\log\frac{1}{\varepsilon}$, and $\log U$ are all bounded by $U$. Consequently,
\begin{equation}\label{eq:explicit-first-term}
 \frac{\eta^{-3}ra\log\left(e+\frac{m}{a}\right)}{m}
 =\frac{\eta^{-3}\log(e+rx)}{x}
 \le\frac{C(1+\log B)}{B}\varepsilon^2.
\end{equation}
The remaining terms in \eqref{eq:Km-explicit}, using \eqref{eq:level-probability-cost}, satisfy
\[
 \frac{\eta^{-1}r\log\frac{L}{\rho}}{m}
 \le\frac{C_{A_0}}{B}\varepsilon^4,\qquad
 \frac{\eta^{-2}r}{m}\le\frac{C}{B}\frac{\varepsilon^3}{aU}.
\]
Choosing $B$ sufficiently large therefore ensures $K_m\le c_1\varepsilon^2m$, as required in \eqref{eq:absorption-hyp}, and increasing it further ensures both $m\ge\frac{64r}{\varepsilon^2}$ and $m\ge C_{A_0}\log n$. The one-target concentration estimate and witness argument in \S\ref{sec:completion} now prove the Bernoulli assertion with failure probability at most $Cn^{-A_0}$; for prescribed cardinality, perform this calculation with exponent $A_0+1$ and use \eqref{eq:binomial-conditioning-cost}, absorbing the change into $C_{A_0}$.
\end{proof}

For fixed $\varepsilon$, the logarithm in \eqref{eq:explicit-row-count} is bounded by $C_\varepsilon L_r$, so the explicit estimate agrees with the fixed-distortion order. Its fifth power of $\varepsilon^{-1}$ reflects the losses identified above, whereas sharp distortion dependence would require a further improvement in the argument.

The explicit estimate separates the dependence on distortion from the logarithmic factors in the row count, which leaves the question of whether those factors are necessary at fixed distortion. We next compare with the Walsh obstruction, whose sparse kernel vectors identify the parameter range in which the sampling theorem has the correct order.

\subsection{Walsh matrices and the matching lower bound}
\label{sec:walsh}

To compare the row count with the known obstruction, put $n=2^d$ and index rows and columns by $\mathbb F_2^d$, defining
\begin{equation}\label{eq:walsh-matrix}
 W_{u,v}:=(-1)^{u\cdot v}\qquad(u,v\in\mathbb F_2^d).
\end{equation}
The dot product is computed modulo two, so summation of a nontrivial character implies $W^TW=nI$ and places \eqref{eq:walsh-matrix} within Theorem~\ref{thm:main}.

\begin{theorem}[B\l asiok--Lopatto--Luh--Marcinek--Rao]\label{thm:walsh-lower}
There are absolute constants $c,C_*>0$ such that, if $n=2^d$, $r=2^k$, and
\begin{equation}\label{eq:walsh-parameter-range}
 \min\{k,d-k\}\ge C_*\log(2d),
\end{equation}
then independent selection of the rows of \eqref{eq:walsh-matrix} with probability $p$ satisfying
\begin{equation}\label{eq:walsh-lower-count}
 pn\le cr\log r\log\frac{n}{r}
\end{equation}
produces a matrix with a nonzero $r$-sparse vector in its kernel with probability tending to one as $d\to\infty$.
\end{theorem}

This is \cite[Theorem 3.1]{BLLMR}, with the constant in \eqref{eq:walsh-parameter-range} absorbing the choice of logarithm base. Such a kernel vector rules out RIP with any constant less than one, so comparison with \eqref{eq:m-range} identifies the Bernoulli threshold up to constants in the stated parameter range. For example, if $0<\theta<1$ is fixed and $k=\lfloor\theta d\rfloor$, then $r\approx n^\theta$ and
\begin{equation}\label{eq:polynomial-sparsity-count}
 r\log\frac{2en}{r}\log(2r)\approx_\theta r(\log n)^2,
\end{equation}
so the logarithm saved from the earlier sufficient count brings the upper bound to the Walsh obstruction in this regime.

Having compared the sufficient sampling count with a lower bound for RIP, we next return to its recovery consequences, this time using support information in addition to coefficient norms. The kernel comparison will yield exact recovery of sparse vectors, while the full restricted-isometry estimate will also control observation noise and the error of sparse approximation.

\subsection{Stable sparse recovery}
\label{sec:recovery}

The full norm condition was used in \S\ref{subsec:fourier-ratio} to obtain approximate recovery controlled by the Fourier ratio, while its sparse specialization also permits exact recovery and stability in terms of best sparse approximation. We include the classical deterministic argument from \cite{CRT,CandesTao,KashinTemlyakov} to make the significance of RIP explicit, with the improved measurement count supplied by Theorem~\ref{thm:complex-main}.

For a matrix satisfying \eqref{eq:complex-matrix-assumptions}, define
\begin{equation}\label{eq:unitary-synthesis}
 U:=n^{-\frac{1}{2}}H,\qquad U_S:=\text{the restriction of }U\text{ to the rows in }S.
\end{equation}
The columns $\phi_j:=Ue_j$, where $e_j$ is the $j$th standard coordinate vector, form an orthonormal basis of $\mathbb C^n$. We call this the \emph{synthesis basis} associated with $U$, because $U$ reconstructs a signal from its coefficients:
\[
 f=\sum_{j=1}^n c_j\phi_j=Uc,\qquad c=U^*f.
\]
Thus $U$ is the synthesis matrix, while $U^*$ computes the coefficients in this basis. The basis is \emph{flat}, meaning that $|\phi_j(i)|=\frac{1}{\sqrt n}$ for every $i,j$, since \eqref{eq:complex-matrix-assumptions} forces $|H_{ij}|=1$. For the inverse Fourier matrix in \eqref{eq:unitary-fourier}, these basis vectors are the normalized characters and $c=\widehat f$.

Because $U$ is unitary, a vector is invisible to the retained measurements precisely when its image under $U$ is supported on $S^c$, which yields
\begin{equation}\label{eq:missing-span}
 \ker U_S=\operatorname{span}\{U^*e_i:i\in S^c\}.
\end{equation}
Indeed, $h=U^*(Uh)$ proves one inclusion and $U(U^*e_i)=e_i$ proves the other, identifying the unobserved span which appears in the subsystem and uncertainty arguments of \cite{IosevichMayeli,IKLM}. A norm comparison on this entire span controls vectors whose mass is concentrated on a small set, as expressed by the following criterion.

\begin{lemma}[A norm comparison implying exact recovery]\label{lem:norm-comparison-recovery}
Suppose that a complex matrix $B$ with $n$ columns satisfies
\begin{equation}\label{eq:kernel-norm-comparison}
 \|h\|_2\le\alpha\|h\|_1\qquad(h\in\ker B),
\end{equation}
and that $1\le s\le n$ with $\alpha\sqrt s<\frac{1}{2}$. Then every $s$-sparse $x$ is the unique minimizer of
\begin{equation}\label{eq:exact-basis-pursuit}
 \min\{\|v\|_1:Bv=Bx\}.
\end{equation}
\end{lemma}

\begin{proof}
Let $T=\operatorname{supp}x$, and write $h_T$ for the restriction of a vector $h$ to $T$, extended by zero on its complement. For nonzero $h\in\ker B$, Cauchy--Schwarz and \eqref{eq:kernel-norm-comparison} imply
\[
 \|h_T\|_1\le\sqrt s\|h\|_2
 \le\alpha\sqrt s\|h\|_1<\frac{1}{2}\|h\|_1,
\]
so $\|h_T\|_1<\|h_{T^c}\|_1$. Every feasible vector distinct from $x$ is of the form $x+h$ for such an $h$, and hence
\[
 \|x+h\|_1
 \ge\|x\|_1-\|h_T\|_1+\|h_{T^c}\|_1>\|x\|_1,
\]
which proves uniqueness in \eqref{eq:exact-basis-pursuit}.
\end{proof}

RIP strengthens this kernel argument by controlling errors which are only approximately invisible because of noise, with the row normalization fixed by \eqref{eq:homogeneous-sampling}. For an integer $1\le s\le n$, define the best $s$-term error by
\begin{equation}\label{eq:best-s-term}
 \sigma_s(x)_1:=\inf\{\|x-v\|_1:|\operatorname{supp}v|\le s\},
\end{equation}
which equals the sum of the coordinate magnitudes left after retaining $s$ largest coordinates of $x$.

\begin{proposition}[Sparse recovery from RIP]\label{prop:recovery}
Suppose $3s\le n$ and $B$ has RIP of order $3s$ with constant $0<\delta<\frac{1}{4}$, see \eqref{eq:rip-definition}. For data $b=Bx+e$ with $\|e\|_2\le\nu$, let $x^\sharp$ be any minimizer of
\begin{equation}\label{eq:basis-pursuit}
 \min_{v\in\mathbb C^n}\|v\|_1
 \quad\text{subject to}\quad\|Bv-b\|_2\le\nu.
\end{equation}
Then
\begin{equation}\label{eq:stable-recovery}
 \|x^\sharp-x\|_2\le
 C\left(\frac{\sigma_s(x)_1}{\sqrt s}+\nu\right),
\end{equation}
with $\sigma_s(x)_1$ as in \eqref{eq:best-s-term} and $C$ absolute. In particular, every $s$-sparse vector is recovered exactly when $\nu=0$.
\end{proposition}

\begin{proof}
Existence follows by restricting to the nonempty compact feasible set with $\|v\|_1\le\|x\|_1$, so write $h=x^\sharp-x$ and choose $T_0$ to consist of $s$ largest coordinates of $x$. Minimality and feasibility imply
\begin{equation}\label{eq:cone-tube}
 \|h_{T_0^c}\|_1\le\|h_{T_0}\|_1+2\sigma_s(x)_1,
 \qquad\|Bh\|_2\le2\nu.
\end{equation}
For the first inequality, compare $\|x+h\|_1\le\|x\|_1$ with
\[
 \|x+h\|_1\ge\|x_{T_0}\|_1-\|h_{T_0}\|_1
                    +\|h_{T_0^c}\|_1-\|x_{T_0^c}\|_1,
\]
and use $\|x_{T_0^c}\|_1=\sigma_s(x)_1$. The second inequality follows by subtracting the two feasible measurement vectors.

Partition $T_0^c$ into successive sets $T_1,T_2,\ldots$ of size $s$, except possibly the last, in decreasing order of $|h_k|$. Every coordinate in a nonempty $T_j$, $j\ge2$, is at most $\frac{\|h_{T_{j-1}}\|_1}{s}$, while the preceding block has exactly $s$ elements, so
\[
 \|h_{T_j}\|_2\le\frac{\|h_{T_{j-1}}\|_1}{\sqrt s};
\]
summing over the disjoint preceding blocks yields
\begin{equation}\label{eq:block-tail}
 \sum_{j\ge2}\|h_{T_j}\|_2\le\frac{\|h_{T_0^c}\|_1}{\sqrt s}.
\end{equation}
Put $T=T_0\cup T_1$, and recall the restricted orthogonality consequence of RIP:
\begin{equation}\label{eq:restricted-orthogonality}
 |\langle Bu,Bv\rangle|\le\delta\|u\|_2\|v\|_2
 \quad\text{if}\quad
 \operatorname{supp}u\cap\operatorname{supp}v=\emptyset,
 \quad|\operatorname{supp}u\cup\operatorname{supp}v|\le3s.
\end{equation}
To verify \eqref{eq:restricted-orthogonality}, restrict $B$ to the union of the two supports: the Hermitian Gram matrix minus the identity has spectrum in $[-\delta,\delta]$ by \eqref{eq:rip-definition}, and $\langle u,v\rangle=0$.

If $h_T=0$, the largest error coordinates outside $T_0$ all vanish, forcing $h=0$, so assume $h_T\ne0$. Decomposing $h=h_T+\sum_{j\ge2}h_{T_j}$ and using \eqref{eq:restricted-orthogonality} and \eqref{eq:block-tail}, we obtain
\begin{align*}
 (1-\delta)\|h_T\|_2^2
 &\le\|Bh_T\|_2^2\\
 &\le|\langle Bh_T,Bh\rangle|
      +\sum_{j\ge2}|\langle Bh_T,Bh_{T_j}\rangle|\\
 &\le\sqrt{1+\delta}\|h_T\|_2\|Bh\|_2
      +\delta\|h_T\|_2\frac{\|h_{T_0^c}\|_1}{\sqrt s}.
\end{align*}
Divide by $\|h_T\|_2$ and use \eqref{eq:cone-tube} together with $\|h_{T_0}\|_1\le\sqrt s\|h_T\|_2$, obtaining
\[
 (1-2\delta)\|h_T\|_2
 \le2\sqrt{1+\delta}\nu+\frac{2\delta\sigma_s(x)_1}{\sqrt s}.
\]
Finally, \eqref{eq:block-tail} and \eqref{eq:cone-tube} imply
\[
 \|h\|_2\le\|h_T\|_2+\frac{\|h_{T_0^c}\|_1}{\sqrt s}
 \le2\|h_T\|_2+\frac{2\sigma_s(x)_1}{\sqrt s},
\]
which proves \eqref{eq:stable-recovery} because $1-2\delta>\frac{1}{2}$.
\end{proof}

\begin{corollary}[Uniform sparse recovery]\label{cor:sampled-recovery}
Fix $A_0>0$, let $H$ satisfy \eqref{eq:complex-matrix-assumptions}, and suppose $3s\le n$. If
\begin{equation}\label{eq:sparse-recovery-row-count}
 C_{A_0}s\log\frac{2en}{s}\log(2s)\le m\le\frac{n}{2},
\end{equation}
then either Bernoulli sampling with expected size $m$, or uniform sampling of $m$ distinct rows when $m$ is an integer, produces a matrix
\begin{equation}\label{eq:recovery-matrix}
 B:=m^{-\frac{1}{2}}H_S
\end{equation}
for which \eqref{eq:stable-recovery} holds simultaneously for every signal and admissible noise vector, with probability at least $1-Cn^{-A_0}$.
\end{corollary}

\begin{proof}
Apply Theorem~\ref{thm:complex-main} with $r=3s$ and RIP constant $\frac{1}{8}$, observing that its sufficient row count is at most a constant multiple of the left threshold in \eqref{eq:sparse-recovery-row-count}. On that single RIP event, Proposition~\ref{prop:recovery} applies to every signal and admissible noise vector, proving the required uniformity.
\end{proof}

The sparse recovery result uses information beyond the norm constraint alone, so we finish by clarifying the scope of the larger class treated by the sampling theorem. Examples with full support, a change of synthesis basis, and normalization by the realized sample size will relate that scope back to the Fourier ratio formulation and distinguish the conclusions that require additional sparsity.

\subsection{Norm classes, basis dependence, and normalization}
\label{subsec:beyond-support}

The ratio in \eqref{eq:norm-ratio} permits full support, a distinction which is relevant when the Fourier coefficients decay without vanishing. For example, define
\begin{equation}\label{eq:dense-vector}
 y:=(1,2^{-1},2^{-2},\ldots,2^{-(n-1)}),\qquad n\ge3,
\end{equation}
so that summing the two geometric series yields
\begin{equation}\label{eq:dense-example}
 \frac{\|y\|_1^2}{\|y\|_2^2}
 =3\frac{1-2^{-n}}{1+2^{-n}}<3.
\end{equation}
In particular, the normalization of \eqref{eq:dense-vector} belongs to $\mathcal K_3$, see \eqref{eq:target-class}, although all $n$ coordinates are nonzero. Permutation and multiplication by coordinatewise complex phases leave \eqref{eq:dense-example} unchanged, so at fixed accuracy the sampling theorem preserves all these vectors, and indeed the entire corresponding norm class, with a sufficient count of order $\log n$ whenever the subsampling range is nonempty.

The argument requires no support estimate: the $\ell^1$ constraint bounds each row amplitude and supplies the convex representation in \eqref{eq:target-representation}, while the $\ell^2$ normalization supplies the entropy improvement through \eqref{eq:coefficient-masses}. Later, the same weighted record bound applies to every target, and the one-target variance calculation uses only the total and maximum energies from \eqref{eq:basic-energy}. This explains why the improved ambient logarithm is available on the whole class rather than only on its sparse vectors.

The distinction between norm preservation and exact reconstruction remains relevant, since differences of two vectors in the class need not satisfy the same norm condition. Proposition~\ref{prop:fourier-ratio-recovery} addresses this issue by separating errors inside the norm class, controlled by the measurements, from errors outside it, controlled by their Fourier $\ell^1$ norm. Corollary~\ref{cor:fourier-ratio-accuracy} then reaches a prescribed relative error by enlarging the threshold; the stronger sparse conclusions in \S\ref{sec:recovery} additionally use support or best-approximation information.

The Fourier formulation also extends to any flat orthonormal synthesis basis, as defined after \eqref{eq:unitary-synthesis}. If $U$ is its synthesis matrix, then $U^*f$ is the coefficient vector of $f$ in that basis, and we define the corresponding coefficient ratio by
\begin{equation}\label{eq:basis-ratio}
 \operatorname{FR}_U(f):=\frac{\|U^*f\|_1}{\|f\|_2}\qquad(f\ne0).
\end{equation}
Since $H(U^*f)=\sqrt n f$, replacing $\widehat f$ by $U^*f$ in the sampling and minimization arguments proves the same conclusions for \eqref{eq:basis-ratio}. This is the orthonormal-basis viewpoint of \cite{BINFourierRatio}, here restricted to the flat bases in our theorem, and the value of the ratio depends on the chosen basis.

There is a final normalization observation which follows from the RIP event itself. Every entry of a matrix in \eqref{eq:complex-matrix-assumptions} has absolute value one, so applying \eqref{eq:homogeneous-sampling} to a coordinate vector implies
\begin{equation}\label{eq:sample-cardinality-control}
 \bigl||S|-m\bigr|\le\varepsilon m.
\end{equation}
Consequently $S\ne\emptyset$ on this event, and division by its actual cardinality yields
\begin{equation}\label{eq:actual-cardinality-normalization}
 \frac{1-\varepsilon}{1+\varepsilon}\|y\|_2^2
 \le\frac{\|H_Sy\|_2^2}{|S|}
 \le\frac{1+\varepsilon}{1-\varepsilon}\|y\|_2^2
 \qquad(\|y\|_1\le\sqrt r\|y\|_2).
\end{equation}
In particular, normalization by expected or actual sample size is available on the same event, with the RIP constant in the latter convention at most $\frac{2\varepsilon}{1-\varepsilon}$.

The remaining quantitative questions include sharp distortion dependence and the endpoint sparsities outside \eqref{eq:walsh-parameter-range}. The calculation in \S\ref{sec:distortion} isolates the powers lost by this proof, while the predictor's present role is to count approximations. An efficient procedure for certifying RIP for a particular sample would require additional ideas.

\section*{Acknowledgments}

Alex Iosevich was partially supported by National Science Foundation grant DMS-2506858 and by a 2026 Simons Fellowship in Mathematics. Ben Krause was partially supported by the Engineering and Physical Sciences Research Council through New Investigator Award EP/W010275/2 and by UK Research and Innovation through Horizon Europe Guarantee grant EP/Y007336/1 for the project PCMEA (Pointwise Convergence of Multiple Ergodic Averages).

Discussions with ChatGPT contributed mathematical suggestions to the development of the common-potential argument and its application in this paper. The authors take responsibility for the mathematical statements and proofs.

\end{document}